\RequirePackage{fix-cm}
\documentclass[smallextended]{svjour3}       %
\smartqed  %
\usepackage{graphicx}

\usepackage{etoolbox}
\usepackage{enumitem}
\usepackage{amsmath,amssymb,array,color,colortbl,graphicx,multirow,mathtools}
\usepackage{microtype}
\usepackage{comment}
\usepackage{todonotes}
\usepackage{threeparttable}
\usepackage{comment}
\usepackage{listings,color}
\usepackage{verbatim}
\usepackage{multirow}
\usepackage{listings}

\usepackage{framed,color}
\usepackage[para]{footmisc}
\newtheorem{clm}{Claim}

\definecolor{mygreen}{rgb}{0,0.5,0}

\newenvironment{cenv}{\begin{list}{}{%
      \setlength{\labelwidth}{1.5em}%
      \setlength{\leftmargin}{\labelwidth}%
      \addtolength{\leftmargin}{\labelsep}%
      \setlength{\listparindent}{0em}%
      \setlength{\topsep}{10pt}%
      \setlength{\itemsep}{5pt}%
      \setlength{\parsep}{0pt}%
    }
  }{
  \end{list}
}

\usepackage{mathrsfs}

\newcommand{\bagsize}{|b|}

\newcounter{claimcounter}

\newenvironment{ClaimProof}[1][]{\noindent{%
\ifthenelse{\equal{#1}{}}{{\sl Proof.\ }}{{\sl #1.\ }}%
}}{\hspace*{1em}\nobreak\hfill$\dashv$\endtrivlist\addvspace{2ex plus
0.5ex minus0.1ex}}

\newcommand{\algcaption}[1]{\item [] \textbf{#1}}
\newcounter{algorithm1counter}
\newenvironment{algorithm}{
	
	~\refstepcounter{algorithm1counter}
	\begin{cenv}
		\item[{\textbf{Algorithm \arabic{algorithm1counter}.}}]
	}{
	\end{cenv}
}

\usepackage{xspace}

\definecolor{saeedgreen}{rgb}{0.0, 0.42, 0.24}

\newcommand{\stefan}[1]{\textit{\textcolor{red}{[stefan]: #1}}} %
\newcommand{\klaus}[1]{\textit{\textcolor{blue}{[klaus]: #1}}} %

\usetikzlibrary{calc}
\usetikzlibrary{chains}
\usetikzlibrary{patterns}
\usetikzlibrary{decorations.pathreplacing}
\usetikzlibrary{er}
\usetikzlibrary{shapes}
\usetikzlibrary{shapes.multipart}
\usetikzlibrary{arrows}
\tikzset{markovstate/.style={shape=circle,draw=black,align=center,inner sep=0,minimum width=7mm}}
\tikzset{markovedge/.style={-latex}}
\tikzset{label/.style={}}
\tikzset{pkt/.style={draw,rectangle,fill=gray!15,minimum height=25pt,align=center,font=\footnotesize}}
\tikzset{pktlen/.style={above=-2pt,font=\scriptsize}}

\definecolor{mycolor}{RGB}{0,200,100}

\usepackage{thmtools}
\usepackage{thm-restate}

\usepackage{amsmath,amssymb}
\let\doendproof\endproof
\renewcommand\endproof{~\hfill\qed\doendproof}

\usepackage{url}

\usepackage{hyperref}

\journalname{arXiv}
\begin{document}

\title{Walking Through Waypoints%
}

\author{Saeed Akhoondian Amiri        \and
        Klaus-Tycho Foerster					\and
				Stefan Schmid
}

\authorrunning{S.~Akhoondian Amiri, K.-T.~Foerster, and S.~Schmid} %

\institute{Saeed Akhoondian Amiri \at
               Max Planck Institute for Informatics (Saarland), Germany\\
              Tel.: +49-681-9325-1012\\
              Fax: +49-681-9325-199\\
							Orcid: \url{https://orcid.org/0000-0002-7402-2662}\\
              \email{samiri@mpi-inf.mpg.de}%
           \and
           Klaus-Tycho Foerster \and Stefan Schmid \at
              University of Vienna, Austria\\
							Tel.: +43-1-4277-786-10/20\\
              Fax: +43-1-4277-878620\\
							Orcid K.-T. Foerster: \url{http://orcid.org/0000-0003-4635-4480}\\
							Orcid Stefan Schmid: \url{http://orcid.org/0000-0002-7798-1711} \\
              \email{\{klaus-tycho.foerster,stefan\_schmid\}@univie.ac.at}\\
}

\date{~}

\maketitle

\begin{abstract}
We initiate the study of a fundamental combinatorial problem:
Given a capacitated graph $G=(V,E)$, find a shortest walk
(``route'')
from a source $s\in V$ to a destination $t\in V$
that includes all vertices specified by a set $\ensuremath{\mathscr{W}}\subseteq V$:
the \emph{waypoints}.
This waypoint routing problem finds immediate applications in the context
of modern networked distributed systems. 
Our main contribution is an exact polynomial-time
algorithm for graphs of bounded treewidth. 
We also show that if the number of waypoints is logarithmically bounded, exact polynomial-time algorithms exist
even for general graphs. 
Our two algorithms provide an almost complete characterization
of what can be solved exactly in polynomial-time:
we show that more general problems (e.g., on grid graphs of maximum degree 3, with slightly more waypoints) are computationally intractable.


\vspace{0.3cm}
\noindent\textbf{Bibliographical note} A preliminary extended abstract appeared in the Proceedings of the 13th Latin American Theoretical Informatics Symposium\\ (LATIN~2018), Springer, 2018~\cite{DBLP:conf/latin/AmiriFS18}.
\end{abstract}

\clearpage
\section{Introduction}\label{sec:intro}

How fast can we find a shortest route, i.e., \emph{walk}, from a source $s$ to a destination
$t$ which visits a given set of waypoints in a graph,
but also respects edge capacities, limiting the number of traversals?
This fundamental combinatorial problem finds immediate applications,
e.g., in modern networked systems connecting distributed network functions
However, surprisingly 
little is known today about the fundamental algorithmic problems
underlying \emph{walks through waypoints}.

The problem features interesting connections to the disjoint
paths problem, however, in contrast to disjoint paths, we (1) consider
\emph{walks} (of unit resource \emph{demand} each time an edge is traversed) 
on \emph{capacitated} graphs rather than \emph{paths} on 
\emph{uncapaciatated} graphs, and we
(2) require that a set of specified vertices are visited. We refer to Figure~\ref{fig:introductory1} for two examples. 

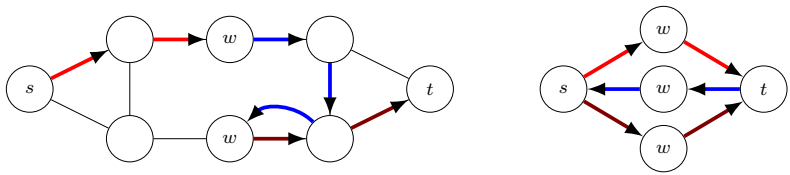
\begin{figure}[h]
	\vspace{-0.1cm}
\centering
 \resizebox{0.88\columnwidth}{!}{ 
	\centering
	\begin{tikzpicture}[auto]
	\node [markovstate] (s) at (0, 0.75) {$s$};
	\node [markovstate] (v11) at (1.5, 0) {};%
	\node [markovstate] (v12) at (3, 0) {$w$};
	\node [markovstate] (v13) at (4.5, 0) {};%
	\node [markovstate] (t) at (6, 0.75) {$t$};
	
	\node [markovstate] (v21) at (1.5, 1.5) {};%
	\node [markovstate] (v22) at (3, 1.5) {$w$};
	\node [markovstate] (v23) at (4.5, 1.5) {};%
	
	\node [markovstate] (s1) at (8, 0.75) {$s$};
	\node [markovstate] (t1) at (11, 0.75) {$t$};
	\node [markovstate] (wm) at (9.5, 0.75) {$w$};
	\node [markovstate] (wt) at (9.5, 1.65) {$w$};
	\node [markovstate] (wb) at (9.5, -0.15) {$w$};
	
	\draw [markovedge, ultra thick, draw=red!200] (s1) to (wt); 
	\draw [markovedge, ultra thick, draw=red!200] (wt) to (t1); 
	\draw [markovedge, ultra thick, draw=blue!200] (t1) to (wm);
	\draw [markovedge, ultra thick, draw=blue!200] (wm) to (s1);
	\draw [markovedge, ultra thick, draw=brown!200] (s1) to (wb);
	\draw [markovedge, ultra thick, draw=brown!200] (wb) to (t1);
	
	\draw (s) to (v11);
	\draw (v11) to (v12);
	\draw (v12) to (v13);
	\draw (v13) to (t);
	
	\draw (s) to (v21);
	\draw (v21) to (v22);
	\draw (v22) to (v23);
	\draw (v23) to (t);
	
	\draw (v11) to (v21);
	\draw (v13) to (v23);
	
	\draw [markovedge, ultra thick, draw=red!200] (s) to (v21);  %
	\draw [markovedge, ultra thick, draw=red!200] (v21) to (v22);
	
	\draw [markovedge, ultra thick, draw=blue!200] (v22) to (v23);
	\draw [markovedge, ultra thick, draw=blue!200] (v23) to (v13);
	\draw [markovedge, ultra thick, draw=blue!200] (v13) [out=135,in=45] to (v12);
	
	\draw [markovedge, ultra thick, draw=brown!200] (v12) to (v13);
	\draw [markovedge, ultra thick, draw=brown!200] (v13) to (t);

	\end{tikzpicture}
	}
	\caption{Two shortest walks and their decompositions into three paths each:
	In both graphs, we walk through all waypoints from $s$ to $t$ by first taking the
	\emph{red}, then the \emph{blue}, and lastly the 
	\emph{brown} path. The existence of a solution 
	in the left graph (e.g., a walk of length 7 in this case) relies on one edge incident to a waypoint having a capacity of at least two. In the right graph, it is sufficient that all edges have unit capacity.  Note that no $s-t$ path through all waypoints exists, 
	for either graph.}
	\label{fig:introductory1}
	\vspace{-0.6cm}
\end{figure}

\subsection{Model}\label{subsec:model}
The inputs to the \emph{Waypoint Routing Problem (WRP\xspace)}
are: 
\begin{enumerate}
	\item A connected, undirected,  
capacitated and weighted graph $G=(V,E,c,\omega)$
consisting of $n=|V|>1$ vertices, 
where $c\colon E\rightarrow \mathbb{N}$ represents edge capacities and $\omega \colon E\rightarrow \mathbb{N}$
represents the polynomial edge costs.
\item A source-destination vertex pair $s,t\subseteq V(G)$.
\item  A set of $k$ waypoints
$\ensuremath{\mathscr{W}}=(w_1,\ldots,w_k) \in V(G)^k$.
\end{enumerate}

We observe that the route (describing a walk) can be decomposed to %
simple paths between terminals and waypoints, and 
we ask: Is there a \emph{route} $R$,
which w.l.o.g.~can be decomposed into $k+1$ path segments
 $R=P_1\oplus\ldots\oplus P_{k+1}$, s.t.%
\begin{enumerate}
\item \emph{Capacities are respected:}
We assume unit demands and require $|\{i\mid e\in P_i \in R, i\in [1,k+1]\}| \le c(e)$ 
for every edge $e\in E$.
\item  \emph{Waypoints are visited:} 
Every element in $\ensuremath{\mathscr{W}}$ appears as an endpoint of exactly two distinct
  paths in route $R$ and $s$ is an endpoint of $P_1$ and
  $t$ is an endpoint of $P_{k+1}$.
	Note that the $k$ waypoints can be visited in any order.
\item \emph{Walks are short:} The length $\ell=|P_1|+\ldots+|P_{k+1}|$
of route $R$ w.r.t.~edge traversal cost $\omega$
is minimal.
\end{enumerate}

\noindent \textbf{Remark I: Reduction to Edge-Disjoint Problems.} Without loss of generality, 
it suffices to consider capacities $c\colon E\rightarrow \left\{1,2\right\}$, as shown in~\cite[Fig.~1]{DBLP:conf/soda/KleinM14}, also stated as Lemma~\ref{lemma:twice} in the Appendix:
 a walk $R$ which traverses an edge $e$ 
more than twice, cannot be a shortest one.

This also gives us a simple reduction of the capacitated problem
to an uncapacitated (i.e., unit capacity), \emph{edge-disjoint} problem variant, by using at most two parallel edges
per original edge. Depending
on the requirements, we will further subdivide these parallel edges into paths (while preserving distances
and graph properties such as treewidth, at least approximately).

\vspace{0.2cm}

\noindent \textbf{Remark II: Reduction to Cycles.} 
Without loss of generality and to simplify presentation, we focus on the 
special case $s=t$. 
In the Appendix (Lemma~\ref{lemma:s=t}), we show %
that 
we can modify instances with $s \neq t$ to instances
with $s=t$ in a distance-preserving manner and by increasing the treewidth by at most one.
Our NP-hardness results hold for $s=t$ as well.

\vspace{-0.2cm}

\subsection{Our Contributions}

We initiate the study of a fundamental 
waypoint routing problem.
We present polynomial-time algorithms to 
compute shortest routes (walks) through arbitrary waypoints
on graphs of bounded treewidth
and to compute shortest routes on general graphs through a
bounded (but not necessarily constant) number of 
waypoints.
We show that it is hard to significantly generalize 
these results both in terms of the family of graphs as well as in terms of  the number of waypoints, 
by deriving NP-hardness results: Our 
exact algorithms cover a good fraction of the problem space for which
polynomial-time solutions exist.
 More precisely, we present the following results:

\begin{enumerate}

\item \textbf{Shortest Walks on Arbitrary Waypoints:} 
While many vertex disjoint problem variants like Hamiltonian path, TSP, vertex disjoint paths, etc.~are 
often  polynomial-time solvable in graphs of bounded treewidth,
their edge-disjoint counterparts 
are sometimes NP-hard already on series-parallel graphs. 
As the Waypoint Routing Problem is an edge-based problem, 
one might 
expect that the problem is NP-hard already on bounded
treewidth graphs, similarly to the 
edge-disjoint paths problem. 

Yet, and perhaps surprisingly, we prove that a shortest walk through 
an arbitrary number of waypoints can be computed in polynomial time on graphs of bounded treewidth.
By employing a simple trick, we transform the capacitated problem variant to an 
uncapacitated edge-disjoint problem: the resulting uncapacitated graph has 
almost the same treewidth.
We then employ a well-known dynamic programming technique 
on a nice tree decomposition of the graph. 
However, since the walk is allowed to visit a 
vertex multiple times, we cannot rely on techniques which are known for vertex-disjoint paths.
Moreover, we cannot simply use the line graph of the original graph: the resulting graph does not preserve the bounded treewidth property.
Accordingly, 
we develop new methods and tools to deal with these issues.

\item \textbf{Shortest Walks on Arbitrary Graphs:} 
\begin{comment}
We show that a shortest route through a logarithmic
resp.~a route through a loglog number of waypoints can be computed in %
randomized resp.~deterministic polynomial
time on general graphs, by reduction to the vertex-disjoint cycle problems in~\cite{thore-soda,DBLP:conf/ipco/Kawarabayashi08}.
\end{comment}
We show that a shortest route through a logarithmic number of waypoints can be computed in randomized time on general graphs, by reduction to the vertex-disjoint cycle problem in~\cite{thore-soda}. 
Similarly, we show that a route through a loglog number of waypoints can be computed in deterministic polynomial time on general graphs via~\cite{DBLP:conf/ipco/Kawarabayashi08}.

Again, we show that that this is almost tight, in the sense that 
the problem becomes NP-hard for any polynomial number of waypoints.
This reduction shows that the edge-disjoint paths problem is not harder than
the vertex-disjoint problem on general graphs, 
and the hardness result also implies that~\cite{thore-soda} is nearly asymptotically tight
in the number of waypoints. %
\vspace{-0.3cm}
\end{enumerate}

\subsection{A Practical Motivation}

The problem of finding routes through waypoints or specified
vertices is a natural and fundamental one. 
We sketch just one motivating application, arising
in the context of modern networked systems.  
Whereas traditional computer networks were designed with an ``end-to-end principle''~\cite{e2e}
philosophy in mind,
modern networks host an increasing number of ``middleboxes'' or network functions, 
distributed across the network, in order to improve performance (e.g., traffic optimizers, caches, etc.),
security (e.g., firewalls, intrusion detection systems), or scalability (e.g., network
address translation).  
Middleboxes are increasingly virtualized (a trend
known as network function virtualization~\cite{etsi}) and can be deployed 
flexibly at arbitrary locations in the network (not only at the edge) and at low costs.
This requires more flexible routing schemes, e.g., leveraging
software-defined network technology~\cite{road}, to route
the traffic through these (virtualized) middleboxes to compose more complex
network services (also known as service chains~\cite{ETSI1}). 
Thus, 
the resulting traffic routes, through capacitated network links, can be modeled as walks, and the problem of finding
shortest routes through 
such 
middleboxes (the waypoints) is an instance
of WRP\xspace.

\vspace{-0.1cm}

\subsection{Related Work}

The Waypoint Routing Problem is closely related 
to disjoint paths problems arising in many applications~\cite{korte1990paths,ogier1993distributed,srinivas2005finding}.
    Indeed, assuming unit edge capacities
and a single waypoint $w$, the problem of finding
a shortest walk $(s,w,t)$ can be seen as a problem of finding
two shortest (edge-)disjoint paths $(s,w)$ and $(w,t)$ with a common
vertex $w$. More generally, a shortest walk $(s,w_1,\ldots,w_k,t)$
in a unit-capacity graph can be seen as a sequence of $k+1$ disjoint paths.
The edge-disjoint and vertex-disjoint paths problem
(sometimes called \emph{min-sum} disjoint paths)
is a deep and intensively
studied 
combinatorial problem, also in the context of parallel algorithms~\cite{DBLP:journals/siamcomp/KhullerMV92,DBLP:journals/siamcomp/KhullerS91}. Today, we have a fairly good understanding
of the \emph{feasibility} of $k$-disjoint paths:
for constant $k$, polynomial-time algorithms for general graphs
have been found by Ohtsuki~\cite{ohtsuki1981two11}, %
Seymour~\cite{seymour1980disjoint12}, 
Shiloah~\cite{Shiloach:1980:PSU:322203.322207}, %
 and Thomassen~\cite{thomassen1980214}
 in the 1980s, and for general $k$ it is NP-hard~\cite{karp1975computational},
 already on series-parallel graphs~\cite{Nishizeki2001177},
i.e., graphs of treewidth at most two. 

However, 
the \emph{optimization} problem (i.e., finding \emph{shortest} paths)
continues to puzzle researchers, even for $k=2$. 
Until recently, despite the progress on polynomial-time algoritms
for special graph families like variants of planar graphs~\cite{DBLP:conf/csr/AmiriGKS14,verdiere2011shortest3,kobayashi2009shortest7}
or graphs of bounded treewidth~\cite{scheffler1994practical},
no subexponential time algorithm was known
even for the 2-disjoint paths problem on general graphs~\cite{eilam1998disjoint,minsum2,kobayashi2009shortest7}.
A recent breakthrough result shows that optimal solutions can at least be computed in 
\emph{randomized} polynomial time~\cite{thore-icalp}; however,
we still have no \emph{deterministic} polynomial-time algorithm.
Both existing feasible and optimal algorithms 
are often
impractical~\cite{thore-icalp,cygan2013planar,DBLP:journals/siamcomp/Schrijver94,seymour1980disjoint12},
and come with high time complexity.
We also note that there are results on the \emph{min-max} version of the
disjoint paths problem, which asks to minimize the length of the longest path.
The  \emph{min-max} problem is believed to be harder than \emph{min-sum}~\cite{itai1982complexity,kobayashi2009shortest7}.

The problem of finding shortest (edge- and vertex-disjoint)
paths and cycles through $k$ waypoints has
been studied in different contexts already. 
The cycle problem variant is also known as the 
\emph{$k$-Cycle Problem}
and has been a central topic of graph theory since the 1960’s~\cite{DBLP:journals/ijfcs/PerkovicR00}.
A cycle from $s$ through $k=1$
waypoints back to $t=s$ can be found efficiently by
breadth first search, for $k=2$ the problem 
corresponds to finding a integer
flow of size 2 between
two vertices, and for $k = 3$, it can still be solved in linear time~\cite{DBLP:journals/ipl/FleischnerW92,DBLP:journals/tcs/FortuneHW80};
a polynomial-time solution for any constant $k$ follows from the work 
on the disjoint paths problem~\cite{DBLP:journals/jct/RobertsonS95b}. 
The best known deterministic algorithm to compute \emph{feasible} 
(but not necessarily shortest) paths is by 
Kawarabayashi~\cite{DBLP:conf/ipco/Kawarabayashi08}: it finds a cycle for up to $k = O((\log \log n)^{1/10})$
waypoints in deterministic polynomial time. 
Bj\"orklund et al.~\cite{thore-soda} presented a randomized algorithm 
based on algebraic techniques
which finds a shortest simple cycle through a given set of $k$ vertices or edges in an 
$n$-vertex undirected graph in time $2^kn^{O(1)}$.
In contrast 
, we assume capacitated networks and do not enforce routes
to be edge or vertex disjoint, but rather consider (shortest) walks.

For 
capacitated graphs, 
researchers have explored the admission control 
variant:
the problem of admitting a maximal number of routing requests
such that capacity constraints are met. 
Chekuri et al.~\cite{chekuri2009note} 
and Ene et al.~\cite{ene_et_al:LIPIcs:2016:6037} presented
approximation algorithms for maximizing the benefit of admitting 
disjoint paths in bounded treewidth 
graphs with both edge
and vertex capacities. 
Even et al.~\cite{sss16moti,sirocco16path} and Rost et al.~\cite{rost2016service} 
initiated the study
of approximation algorithms for admitting a maximal number of
routing \emph{walks} through waypoints. 
In contrast, we 
focus on the optimal routing of a single walk.

In the context of capacitated graphs and single walks, the applicability of edge-disjoint paths algorithms to the so-called \emph{ordered} Waypoint Routing problem was studied in~\cite{ifip-waypoint,ccr-waypoint}, where the task is to find $k+1$ capacity-respecting paths $(s,w_1,), (w_1,w_2),\dots, (w_k,t)$.
An extension of their methods to the \emph{unordered} Waypoint Routing problem via testing all possible $k!$ orderings falls short of our results: For general graphs, only $O(1)$ waypoints can be considered, and for graphs of bounded treewidth, only $O(\log n)$ waypoints can be routed in polynomial time~\cite{ifip-waypoint,ccr-waypoint};
both results concern feasibility only, but not shortest routes.
We provide algorithms for $O(\log n)$ waypoints on general graphs and $O(n)$ waypoints in graphs of bounded treewidth, 
for shortest routes.

Lastly, for the case that all edges have a capacity of at least two and $s=t$, a direct connection of WRP\xspace to the subset traveling salesman problem (TSP) can be made~\cite{Foerster2017}. 
In the subset TSP, the task is to find a shortest closed walk that visits a given subset of the vertices~\cite{DBLP:conf/soda/KleinM14}.
As optimal routes for WRP\xspace and subset TSP traverse every edge at most twice, optimal solutions for both are identical when $\forall e\in E: c(e)\geq 2$. %
Hence, we can make use of the subset TSP results of Klein and Marx, with time of $(2^{O\left(\sqrt{k}\log k\right)}+\max_{\forall e \in E}{\omega(e)}) \cdot n^{O(1)}$ on planar graphs. 
Klein and Marx also point out applicability of the dynamic programming techniques of Bellman and of Held and Karp, %
 allowing subset TSP to be solved in time of $2^{k} \cdot n^{O(1)}$. %
For a PTAS on bounded genus graphs, we refer to~\cite{DBLP:journals/algorithmica/BorradaileDT14}.
We would like to note 
 that the technique for $s\neq t$ of
Remark~II
does not apply if all edges must have a capacity of at least two. Similarly, it is in general not clear how to directly transfer $s=t$ TSP results to the case of $s\neq t$, cf.~\cite{DBLP:conf/focs/SeboZ16}. 
\begin{comment}
We would like to note again that all the techniques in this paragraph only apply if every edge has a capacity of at least two and if $s=t$. and do not work for edges with capacity one or $s \neq t$. In fact, it is not even clear how to directly transfer $s=t$ TSP results to the case of $s \neq t$, cf.~\cite{DBLP:conf/focs/SeboZ16}.
\end{comment}
Notwithstanding, as WRP\xspace also allows for unit capacity edges (to which subset TSP is oblivious), WRP\xspace is a generalization of subset TSP. %

\begin{comment}
\stefan{to be added somewhere:}
We will make use of a result
from~\cite{DBLP:journals/jct/RobertsonS95b}, showing that for any
constant $k'$ and $s_1-t_{1}$, $\dots$, $s_{k'}-t_{k'}$ , the
$k'$-vertex-disjoint path problem can be decided in polynomial-runtime
of $O(n^3)$ on undirected graphs $G'$, later improved by
Perkovi{\'{c}} and Reed~\cite{DBLP:journals/ijfcs/PerkovicR00} to
$O(n^2)$. \klaus{Saeed might know better stuff here} \klaus{?todo: some note regarding transformation from vertex to edge-disjoint for disjoint paths?} 
\stefan{where exactly do we need it? perhaps we can then introduce it
in the technical parts rather than the related work?}
\end{comment}

\subsection{Paper Organization}
In Section~\ref{sec:algorithm} we present our results for bounded treewidth graphs
and Section~\ref{sec:polylog} considers general graphs.
We derive distinct NP-hardness results in Section~\ref{sec:nphard}
and conclude in Section~\ref{sec:future}.
In order to improve presentation,
some technical contents are deferred to the Appendix.

\section{Walking Through Waypoints on Bounded Treewidth}\label{sec:algorithm} %

The complexity of the Waypoint Routing Problem on bounded treewidth
graphs is of particular interest: while vertex-disjoint paths and cycles 
problems are often polynomial-time
solvable on bounded treewidth graphs (e.g., vertex disjoint paths~\cite{DBLP:journals/jct/RobertsonS95b}, 
vertex coloring,
Hamiltonian cycles~\cite{DBLP:journals/dam/ArnborgP89}, Traveling Salesman~\cite{DBLP:journals/iandc/BodlaenderCKN15}, see also~\cite{DBLP:journals/actaC/Bodlaender93,fellows2007complexity})
many edge-disjoint problem variants are NP-hard
(e.g., edge-disjoint paths~\cite{Nishizeki2001177}, edge coloring~\cite{DBLP:journals/ipl/Marx04}).
Moreover, the usual \emph{line graph} construction approaches
to transform vertex-disjoint to edge-disjoint
problems are not applicable as such transformations do not preserve bounded treewidth. %

Against this backdrop, we show that indeed shortest routes through arbitrary 
waypoints can be computed
in polynomial-time for bounded treewidth graphs.

\begin{theorem}\label{thm:everything}
The Waypoint Routing Problem can be solved in time of $n^{O(\texttt{tw}^2)}$, where $\texttt{tw}$ denotes the treewidth
of the graph $G=(V,E)$ with $|V|=n$ vertices. %
\end{theorem}

In other words, the Waypoint Routing Problem is in the complexity class XP~\cite{DBLP:books/sp/CyganFKLMPPS15,DBLP:series/mcs/DowneyF99} w.r.t.~treewidth. We obtain:

\begin{corollary}\label{corr:how-far-can-we-go}
The Waypoint Routing Problem can be solved in polynomial time for graphs of bounded treewidth $\texttt{tw} \in O(1)$.
\end{corollary}

\noindent\textbf{Overview.} 
We 
describe our algorithm in terms of a
\emph{nice tree decomposition}~\cite[Def.~13.1.4]{DBLP:books/sp/Kloks94} (\S\ref{subsec:tw-prelim}).
We 
transform the edge-capacitated problem into an 
edge-disjoint problem (with unit edge capacities~\S\ref{subsec:unified}),
leveraging a simple observation on 
the structure of waypoint walks and preserving distances. 
We show that this transformation changes the treewidth by
at most an additive constant.
We then define the separator signatures (\S\ref{subsec:bigproof}) and describe how to inductively generate valid signatures in a bottom 
up manner on the nice tree decomposition, applying the \emph{forget},
\emph{join} and \emph{introduce} operations~\cite[Def.~13.1.5]{DBLP:books/sp/Kloks94} (\S\ref{subsec:prntt}).

The correctness of our approach relies on a crucial observation 
on the underlying Eulerian properties of the Waypoint Routing Problem in Lemma~\ref{lem:eulerianseparator-new},
allowing us to bound the number of partial walks we need to consider
at the separator, see Figure~\ref{fig:Eulerian-walks} for an example.
Finally in \S\ref{sec:all-together}, we bring together the different bits and pieces, and sketch
how to dynamically program~\cite{bodlaender1988dynamic} the shortest waypoint walk on the rooted
separator tree.

\begin{figure}[thbp]
\resizebox{\columnwidth}{!}{ 
	\centering
	\begin{tikzpicture}[auto]
	\node [markovstate] (s1) at (0, 1) {$s_1$};
	\node [markovstate] (s2) at (0, 2.5) {$s_2$};
	\node [markovstate] (s3) at (0, 4) {$s_3$};
	
	\node [markovstate] (a1) at (-2, 1.75) {$a_1$};
	\node [markovstate] (a2) at (-2, 3.25) {$a_2$};
	
	\node [markovstate] (b1) at (2, 1) {$b_1$};
	\node [markovstate] (b2) at (2, 2.5) {$b_2$};
	\node [markovstate] (b3) at (2, 4) {$b_3$};
	
	\draw [dotted] (0,2.5) ellipse (0.8cm and 2.1cm);
	\node (S) at (0,0) {$s$};
	\node (A) at (-2,0) {$A$};
	\node (B) at (2,0) {$B$};

	\draw (s1) to (a1);
	\draw (a1) to (s2);
	\draw (s2) to (a2);
	\draw (a2) to (s3);
	
	\draw (s1) to (s2);
	
	\draw (b1) to (b2);
	\draw (b2) to (b3);
	
	\draw (s1) to (b1);
	\draw (s1) to (b2);
	\draw (s2) to (b2);
	\draw (s3) to (b3);
	
	\draw [markovedge, ultra thick, draw=brown!200] (s2)  to  node[left] {1}(s1);
	\draw [markovedge, ultra thick, draw=green!200] (s1) to node[below] {5}(a1);
	\draw [markovedge, ultra thick, draw=green!200] (a1) to node[below] {6} (s2);
	
	\draw [markovedge, ultra thick, draw=blue!200] (a2) to node[above] {11} (s2);
	\draw [markovedge, ultra thick, draw=blue!200] (s3) to node[above] {10} (a2);
	
	\draw [markovedge, ultra thick, draw=red!200] (b3) to node[above] {9} (s3);
	\draw [markovedge, ultra thick, draw=red!200] (b2) to  node[right] {8}(b3);
	\draw [markovedge, ultra thick, draw=red!200] (s2) to node[above] {7} (b2);
	
	\draw [markovedge, ultra thick, draw=brown!200] (s1) to node[below] {2}(b2);
	\draw [markovedge, ultra thick, draw=brown!200] (b2) to node[right] {3}(b1);
	\draw [markovedge, ultra thick, draw=brown!200] (b1) to node[below] {4}(s1);
	\node [markovstate] (1s1) at (8, 1) {$s_1$};
	\node [markovstate] (1s2) at (8, 2.5) {$s_2$};
	\node [markovstate] (1s3) at (8, 4) {$s_3$};
	
	\node [markovstate] (1a1) at (6, 1.75) {$a_1$};
	\node [markovstate] (1a2) at (6, 3.25) {$a_2$};
	
	\node [markovstate] (1b1) at (10, 1) {$b_1$};
	\node [markovstate] (1b2) at (10, 2.5) {$b_2$};
	\node [markovstate] (1b3) at (10, 4) {$b_3$};
	
	\draw [dotted] (8,2.5) ellipse (0.8cm and 2.1cm);
	\node (1S) at (8,0) {$s$};
	\node (1A) at (6,0) {$A$};
	\node (1B) at (10,0) {$B$};

	\draw [markovedge, ultra thick, draw=red!200] (1s1) to node[below] {1} (1b1);
	\draw [markovedge, ultra thick, draw=red!200] (1b1) to node[right] {2} (1b2);
	\draw [markovedge, ultra thick, draw=red!200] (1b2) to node[below] {3} (1s1);
	\draw [markovedge, ultra thick, draw=red!200] (1s1) to node[right] {4} (1s2);
	\draw [markovedge, ultra thick, draw=red!200] (1s2) to node[above] {5} (1b2);
	\draw [markovedge, ultra thick, draw=red!200] (1b2) to node[right] {6} (1b3);
	\draw [markovedge, ultra thick, draw=red!200] (1b3) to node[above] {7} (1s3);
	
	\draw [markovedge, ultra thick, draw=blue!200] (1s3) to node[above] {8} (1a2);
	\draw [markovedge, ultra thick, draw=blue!200] (1a2) to node[above] {9} (1s2);
	\draw [markovedge, ultra thick, draw=blue!200] (1s2) to node[below] {10} (1a1);
	\draw [markovedge, ultra thick, draw=blue!200] (1a1) to node[below] {11} (1s1);

	\end{tikzpicture}
}
\vspace{0.2cm}
	\caption{Two different methods to choose an Eulerian walk, where the numbers from $1$ to $11$ describe the order of the traversal. In the left walk, the separator $s$ is crossed $4$ times, but only $2$ times in the right walk. Furthermore, in the left walk, there are 2 walks each in $G[A]$ (\textcolor[rgb]{0,1,0}{green} and \textcolor[rgb]{0,0,1}{blue}) and $G[B]$ (\textcolor[rgb]{0.54,0.27,0.07}{brown} and \textcolor[rgb]{1,0,0}{red}), respectively. In the right walk, there is only 1 walk for $G[A]$ (\textcolor[rgb]{0,0,1}{blue}) and 1 walk for $G[B]$ (\textcolor[rgb]{1,0,0}{red}).
	}
	\label{fig:Eulerian-walks}
\end{figure}

\subsection{Treewidth Preliminaries}\label{subsec:tw-prelim}

A \emph{tree decomposition} $\mathcal{T}=(T,X)$ of a graph $G$ consists of a bijection between a tree $T$ and a collection $X$, where every element of $X$ is a set of vertices of $G$ such that:
(1) each graph vertex is contained in at least one tree node (the \emph{bag}
or \emph{separator}),
(2) the tree nodes containing a vertex $v$ form a connected subtree of $T$,
and (3) vertices are adjacent in the graph only when the corresponding subtrees have a node in common.

The \emph{width} of $\mathcal{T}=(T,X)$ is the size of the largest set in $X$ minus 1, with the \emph{treewidth} of $G$ being the minimum width of all possible tree decompositions. %

A \emph{nice tree decomposition} is a tree decomposition such that: 
(1) it is rooted at some vertex $r$,
(2) leaf nodes are mapped to bags of size 1,
and (3) inner nodes are of one of three types:
 \emph{forget} (a vertex leaves the bag in the parent node),
\emph{join}  (two bags defined over the same vertices
are merged) and \emph{introduce} (a vertex is added to the bag in the parent node).
The tree can be iteratively constructed by applying simple \emph{forget},
\emph{join} and \emph{introduce} types.

Let $b \in X$ be a bag of the decomposition corresponding to a
vertex $b \in V(T)$. We denote by $T_b$ the maximal subtree of $T$ which is
rooted at bag $b$. %
By $G[b]$ we denote the subgraph of
$G$ induced on the vertices in the bag $b$ and by $G[T_b]$ we denote
the subgraph of $G$ which is induced on vertices in all bags in $V(T_b)$. 
We will henceforth assume that a nice tree decomposition
$\mathcal{T}=(T,X)$ of 
$G$ is given, covering its computation in the final steps of the proof of Theorem~\ref{thm:everything}. %

\subsection{Unified Graphs}\label{subsec:unified}

We begin by transforming our graphs into graphs
of unit edge capacity, preserving distances and approximately
preserving treewidth.

\begin{definition}[Unification]
Let $G$ be an arbitrary, edge capacitated graph. 
The \emph{unified} graph $G^u$ of
$G$ is obtained from $G$ by the following operations on
each edge $e\in E(G)$: We replace $e$ by $c(e)$ parallel edges 
$e_1,\ldots,e_{c(e)}$, subdivide each resulting parallel edge
by creating vertices $v_i^e,i\in
[c(e)]$), and set the weight of each subdivided edge to $w(e)/2$
(i.e., the total weight is preserved). We set all
edge capacities in the unified graph to $1$.
Similarly, given the original problem instance $I$ 
of the Waypoint Routing Problem, the  \emph{unified instance} $I^u$
is obtained by replacing the graph
$G$ in $I$ with the graph $G^u$ in $I^u$, without changing the waypoints, the source 
and the destination.
\end{definition}

It follows directly from the construction that 
$I$ and $I^u$ are equivalent with regards to the contained
walks. Moreover, as we will see, the unification process approximately
preserves the treewidth. 
Thus, in the following, we will focus on $G^u$ and $I^u$
only, and implictly assume that $G$ and $I$ are unified.
Before we proceed further, however,
let us introduce some more definitions.
Using Remark~I,
w.l.o.g., we can focus on graphs where for all $e\in E$, 
$c(e)\leq 2$. The treewidth of $G$
and $G^u$ are preserved up to an additive constant.
\begin{lemma}\label{twapprox}
Let $G$ be an edge capacitated graph s.t.~each edge has capacity
at most $2$ and let $\texttt{tw}$ be the treewidth of $G$. Then 
$G^u$ has
treewidth at most $\texttt{tw}+1$.
\end{lemma}

\begin{proof}
Let $\mathcal{T}=(T,X)$ be an optimal tree decomposition of $G$ of width
$\texttt{tw}$. Let $X_b$ be an arbitrary bag of $\mathcal{T}$. We construct the tree
decomposition $\mathcal{T}^u$ of $G^u$ based on $\mathcal{T}$ as
follows. First set $\mathcal{T}^u=\mathcal{T}$. For every edge $e$ of
capacity $c$ in $G$, which
has its endpoints in $X_b$, we create $c$ bags $X_b^i$ ($i\in [c]$) and
set $X_b^i=X_b \cup v_i^e, i\in [c]$. Connect all bags $X_b^i$ to
$X_b$ (i.e., $X_b^i$s are new children of $X_b$). This creates a tree decomposition $\mathcal{T}^u$ of width $\leq \texttt{tw}+1$ of $G$. 
\end{proof}

\noindent\textbf{Leveraging Eulerian Properties} A key insight is that we can leverage the Eulerian properties implied by a waypoint route.
In particular, we show that the traversal of a single Eulerian walk 
(e.g., along an optimal solution of WRP\xspace) can be arranged s.t.~it does not traverse a specified separator too often, for which we will later choose the root of the nice tree decomposition.

\begin{lemma}[Eulerian Separation]\label{lem:eulerianseparator-new}
Let $G$ be an Eulerian graph. Let $s$ be an $(A,B)$ separator of order
$|s|$ in $G$. Then there is a set of $\ell \leq 2|s|$ pairwise edge- disjoint walks $\mathcal{W}=\{W_1,\ldots,W_\ell\}$ of $G$ such that
\begin{enumerate}[label=\textbf{\arabic*)}]

\item\label{euler1} For every $W\in \mathcal{W}$, $W$ has both of its endpoints in $A\cap B$.

\item\label{euler2} Every walk $W\in \mathcal{W}$ is entirely either in $G[A]$ (as  $\mathcal W_A$) or
  in $G[B]$ (as  $\mathcal W_B$).

\item\label{euler3} Let $\beta_A$ be the size of the set of vertices used by $\mathcal{W}_A$ as an endpoint in $s$. Then, $\mathcal{W}_A$ contains at most $\beta_A$ walks. Analogously, for $\beta_B$ and $\mathcal{W}_B$.

\item\label{euler4} There is an Eulerian walk $W$ of $G$ such that: $W\coloneqq
  W_1\oplus\ldots\oplus W_\ell$.

\end{enumerate}
\end{lemma} 

\begin{proof}
$G[A]$ and $G[B]$ share the edges in $G[s]$. 
For this proof, we arbitrarily distribute the edges in $G[s]$, resulting in edge-disjoint $G_A',G_B'$, with $V(G_A')=A$ and $V(G_B')=B$.
As only the vertices in $s$ can have odd degree in $G_A'$, we can cover the edges of $G_A'$ with open walks, starting and ending in different vertices in $s$, and closed walks, not necessarily containing vertices of $s$.
If a vertex in $s$ is the start/end of two different walks, we concatenate these walks into one, repeating this process until for all vertices $v$ in $s$ holds: At most one walk starts or ends at $v$.
Next, we recursively join all closed walks into another walk, with which they share some vertex, cf.~\cite{1991X.1}.

As every vertex in $G_A'$ has a path to a vertex in $s$, we have covered all edges in $G_A'$ with $\alpha_{A'}$ (possibly closed) walks $\mathcal W_{A'}$, with $\alpha_{A'} \leq \beta_{A'} \leq |s|$.
However, all remaining closed walks end in the separator and are pairwise vertex-disjoint from all other (possibly, closed) walks.
We perform the same for $G_B'$ and obtain an analogous $\mathcal W_{B'}$ with $\alpha_{B'} \leq \beta_{B'} \leq |s|$ walks. %
Let us inspect the properties of the union of $\mathcal W_{A'}$ and $\mathcal W_{B'}$:
\begin{itemize} %
	\item All walks have their endpoints in $A \cap B$, respecting \textbf{\ref{euler1}}.
	\item All walks are entirely in $G[A']\subseteq G[A]$ or in $G[B']\subseteq G[B]$, respecting \textbf{\ref{euler2}}.
	\item At each $v \in s$, at most one walk each from $\mathcal W_{A'}$ and $\mathcal W_{B'}$ has its endpoint, respecting \textbf{\ref{euler3}}.
	\item There is no certificate yet that the walks respect \textbf{\ref{euler4}}.
\end{itemize}
As thus, we will now alter $\mathcal W_{A'}$ and $\mathcal W_{B'}$ such that their union respects \textbf{\ref{euler4}}. 

W.l.o.g., we start in any walk $W$ in $\mathcal{W}_{A'}$ to create a set of closed walks $\mathcal{W}_{C'}$. We traverse the walk $W$ from some endpoint vertex $v' \in s$ until we reach its other endpoint $v \in s$, possibly $v'=v$.
As there cannot be any other walks in $\mathcal{W}_{A'}$ with endpoints at $v$, and $v$ has even degree, there are three options:

First, if $W$ is a closed walk and $E(W)=E(G)$, we are done.
Second, if $W$ is a closed walk and $E(W) \neq E(G)$, it does not share any vertex with another walk in $\mathcal W_{A'}$. Hence, there must be a walk or walks in $\mathcal{W}_{B'}$ containing $v$. As $v$ has even degree, we have two options: There could be a closed walk $W'$ in $\mathcal{W}_{B'}$ containing $v$. Then, we set both endpoints of $W'$ to $v$, and as $W'$ does not share a vertex with any other walk in 
$\mathcal W_{B'}$, we are done.
Else, there is an open walk $W'$ which just traverses $v$, not having $v$ as an endpoint. 
We then split $W'$ into two open walks at $v$, increasing $\alpha_B$ and $\beta_B$ by one. %

Third, if $W$ is an open walk, there must be an open walk in $\mathcal{W}_{B'}$ whose start- or endpoint is $v$. We iteratively perform these traversals, switching between $\mathcal{W}_{A'}$ and $\mathcal{W}_{B'}$ eventually ending at $v$ again in a closed walk, with every walk having an endpoint in $v$ being traversed. 

We now repeat this closed walk generation, each time starting at some not yet covered walk. Call this set of closed walks $\mathcal{W}_{C'}$.
If $\mathcal{W}_{C'}$ contains only one closed walk, we found a traversal order of the walks in $\mathcal{W}_{A'}$ and $\mathcal{W}_{B'}$ that yields an Eulerian walk, respecting \textbf{\ref{euler4}}, finishing the argument for that case.
Else, as the graph is connected, there must be two edge-disjoint walks $W_{C_1}, W_{C_2} \in \mathcal{W}_{C'}$ that share a vertex $u$, w.l.o.g., in $G[A']$, with $u \in W_{A_1} \in \mathcal W_{A'}$, $W_{A_1}$ being part of  $W_{C_1}$, and $u \in W_{A_2} \in \mathcal W_{A'}$, $W_{A_2}$ being part of $W_{C_2}$. 
Let $W_{A_1}$ have the endpoints $s_1,r_1$ and $W_{A_2}$ have the endpoints $s_2,r_2$.
We now perform the following, not creating any new endpoints: We cut both walks $W_{A_1}, W_{A_2}$ at $u$ into two walks each. These four walks all have an endpoint in $u$, with the other four ones being $s_1,s_2,r_1,r_2$. We now turn them into two walks again: First, concatenate $s_1,u$ with $u,s_2$ as $W_{A_1}$, and then, concatenate $r_1,u$ with $u,r_2$ as $W_{A_2}$.
Now, we can obtain a single closed walk that traverses $W_{C_1}$ and $W_{C_2}$.
Recursively iterating this process, we obtain a single closed Eulerian walk, respecting \textbf{\ref{euler1}} to \textbf{\ref{euler4}}.
\end{proof}

\subsection{Signature Generation and Properties}\label{subsec:bigproof}

We next introduce the signatures we use to represent previously computed solutions
to subproblems implied by the separators in the (nice) tree decomposition.
For every possible signature, we will determine whether it represents a valid solution
for the subproblem, 
and if so, we store it along with an exemplary sub-solution of optimal weight.

In a nutshell, the signature describes endpoints of (partial) walks
on each side of the separator. These partial walks hence need to
be iteratively merged, forming signatures of longer walks through the waypoints.

\begin{definition}[Signature]\label{def:signature}
Consider a bag $b\in X$.
A signature $\sigma$ of $b$ ($\sigma_{b}$) is a pair, either containing
\begin{enumerate}
	\item   $1)$ an unordered tuple of pairs of vertices $s_i,r_i \in b$ and $2)$ a subset $E_{b}\subseteq E(G[b])$ with $\sigma_{b}=\left(\left(\left(s_1,r_1\right),\left(s_2,r_2\right),\dots,\left(s_\ell,r_\ell\right)\right),E_{b}\right)$ s.t.~$\ell \leq |b|$, or
	\item $1)$ $\emptyset$ and $2)$ $\emptyset$, with $\sigma_{b}=(\emptyset,\emptyset)$, also called an \emph{empty} signature $\sigma_{b,\emptyset}$.
\end{enumerate}

\end{definition}

Note that in the above definition we may have $s_i=r_i$ for some $i$.
We can now define a valid signature and a sub-solution, where we consider the vertex $s=t$ to be a waypoint.

\begin{definition}[Valid Signature, Sub-Solution]\label{def:valid}
Let $b\in X$ and let either $\sigma_{b}=\left(\left\{\left(s_1,r_1\right),\left(s_2,r_2\right),\dots,\left(s_\ell,r_\ell\right)\right\},E_{b}\right)$ or $\sigma_{b} = \sigma_{b,\emptyset}$ be a signature of $b$. 
$\sigma_{b} \neq \sigma_{{b},\emptyset}$ is called a \emph{valid signature} 
if there is a set of pairwise edge-disjoint walks
$ \mathcal{W}_{\sigma_b}=\{W_1,\ldots,W_\ell\}$ such that:
\begin{enumerate}
\item\label{valid-1} If $W_i$ is an open walk then it has both of its endpoints on  $(s_i,r_i)$, otherwise, $s_i=r_i$ and $s_i \in V(W_i)$. %
\item\label{valid-2} Let $\beta$ be the size of the set of endpoints used by $\sigma_{b}$. Then, it holds that $\beta \geq \ell$. %
\item\label{valid-3} For every waypoint $w \in V(T_{b})$ it holds that $w$ is contained in some walk $W_j, 1\leq j \leq \ell$. %
\item\label{valid-4} Every (pairwise edge-disjoint) walk $W_j \in \mathcal{W}_{\sigma_b}$ only uses vertices from $V(T_{b})$ and only edges from $E(T_{b}) \setminus \overline{E}_{b}$, with $\overline{E}_{b} = E({b}) \setminus E_{b}$. %
 \item\label{valid-5} Every edge $e \in E_{b}$ is used by a walk in $\mathcal{W}_{\sigma_b}$. %
\item\label{valid-6} Among all such sets of $\ell$ walks, $\mathcal{W}_{\sigma_b}$ has minimum total weight. 
\end{enumerate}
Additionally, if for a signature $\sigma_{b} \neq \sigma_{{b},\emptyset}$ there is such
a set $\mathcal{W}_{\sigma_{b}}$ (possibly abbreviated by $\mathcal W_{b}$ if clear from the context), we
say that $\mathcal{W}_{\sigma_{b}}$ is a \emph{valid sub-solution} in $G[T_{b}]$.
For some waypoint contained in $G[T_{b}]$, we call a signature $\sigma_{{b},\emptyset}$ 
 \emph{valid}, if there is one walk $W$ associated with it, s.t.~$W$ traverses all waypoints in $G[T_{b}]$, does not traverse any  vertex in $V({b})$, and among all such walks in $G[T_{b}]$ has minimum weight.
If $G[T_{b}]$ does not contain any waypoints, we call the empty signature $\sigma_{{b},\emptyset}$ valid, if there is no walk associated with it.
\end{definition}

\begin{lemma}[Number of different signatures]\label{lem:sigscount}
There are $2^{O(|b|^2)}$ different signatures for $b \in X$. 
\end{lemma}
\begin{proof}
There are at most $2^{O(|b|^2)}$
ways to choose edges from a graph of order $|b|$. 
To distribute up to $2|b|$ endpoints of walks over $|b|$ vertices, we can temporarily add $2|b|$ empty endpoints, which means that there are at most $(4|b|)^{4|b|}$ possibilities.
The empty signature only adds one further possibility.
\end{proof}

\subsection{Programming the Nice Tree Decomposition}\label{subsec:prntt}

The nice tree decomposition directly gives us a  constructive way to dynamically program WRP\xspace
in a bottom-up manner. We first cover leaf nodes 
in Lemma \ref{lem:leaf}, and then work our way up via forget (Lemma~\ref{lem:forget}), introduce (Lemma~\ref{lem:introduce}), and join (Lemma~\ref{lem:join}) nodes, until eventually the root node is reached.
Along the way, we inductively generate all valid signatures at every node.

\begin{lemma}[Leaf nodes]\label{lem:leaf}
Let $b$ be a leaf node in the nice tree decomposition $\mathcal{T}=(T,X)$. Then, in time $O(1)$ we can find all the valid signatures of $b$.
\end{lemma}
\begin{proof}%
We simply enumerate all possible valid signatures.
As a leaf node only contains one vertex $v$ from the graph, all possible edge sets in the signatures are empty, and we have two options for the pairs: First, none, second, $\left(\left(v,v\right)\right)$. The second option is always valid, but the first (empty) one is only valid when $v$ is not a waypoint.
\end{proof}

\begin{lemma}[Forget nodes]\label{lem:forget}
Let $b$ be a forget node in the nice tree decomposition $\mathcal{T}=(T,X)$, with one child $q = \text{child}(b)$, where we have all valid signatures for $q$.
Then, in time $2^{O(\bagsize^2)}$ %
 we can find all the valid signatures of $b$. 
\end{lemma}
\begin{proof}
Let $v \in G$ be the vertex s.t.~$V(b) \cup \left\{v\right\} = V(q)$.
We create all valid signatures for $b$ as follows:
First, if the empty signature is valid for $q$, it is also valid for $b$.
Second, for $\sigma_b = (\mathcal{P},E_b)$, with some $(s,r)$ pairs $\mathcal{P}$ to be a valid signature for $b$, there needs to be a valid signature $\sigma_q = (\mathcal{P},E_b \cup E')$, where $E'$ is a subset of all edges incident to $v$ from $E(G[q])$.
For correctness, consider the following: In all valid (non-empty) signatures of $b$, all walks need to have their endpoints in $V(b)$. As $v$ can only be reached from vertices in $V(G) \setminus V(T_b)$ via vertices in $V(b)$, any walk between vertices of $V(G) \setminus V(T_b)$ and $v$ must pass $V(b)$, i.e., the corresponding signature of $q$ can be represented as a signature of $b$, with possible additional edges.
Checking every of the $2^{O(\bagsize^2)}$ valid signatures of the child as described can be done in time linear in the signature size $O(|b|^2)$, with $2^{O(\bagsize^2)} \cdot O(|b|^2) \in 2^{O(\bagsize^2)}$.
\end{proof}

\begin{lemma}[Introduce nodes]\label{lem:introduce}
Let $b$ be an introduce node in the nice tree decomposition $\mathcal{T}=(T,X)$, with one child $q = \text{child}(b)$, where we have all valid signatures for $q$.
Then
 we can find all the valid signatures of~$b$ in time $|b|^{O(|b|^2)}$.
\end{lemma}

We will exploit the following property for the proof of Lemma~\ref{lem:introduce}.
\begin{property}\label{pro:edges}
Let $b$ be an \emph{introduce} node, where $q$ is a child of $b$, with $V(b) = V(q) \cup \left\{v\right\}$. Then 
$v$ is not adjacent to any vertex in $V(T_b) \setminus V(b)$.
\end{property}
\begin{proof}
Let $v \in G$ be the vertex s.t.~$V(q) \cup \left\{v\right\} = V(b)$.
Recall that $v$ can only have neighbors in $V(b)$ from $V(T_b)$ (Property~\ref{pro:edges}).
From these valid signatures of $q$, we will then create all valid signatures of $b$.
Thus, we can pick some valid signature $\sigma_b$ (possibly empty), with complete knowledge of a valid sub-solution $\mathcal{W}_b$. 
By showing how to obtain a valid signature $\sigma_q$ with valid-subsolution $\mathcal{W}_q$ from $\sigma_b$ and $\mathcal{W}_b$, the process can then be reversed -- as we know all such $\sigma_q$ and $\mathcal{W}_q$.
From $\sigma_b$ and $\mathcal{W}_b$, we now iteratively build a signature $\sigma'_q$ and $\mathcal{W}_q'$, which in the end will represent $\sigma_q$ and $\mathcal{W}_q$.
A first thought is that by removing all walks from $\mathcal{W}_b$ and $\sigma_b$ that contain $v$, we initialize $\sigma'_q$ and $\mathcal{W}_q'$. 
$\sigma'_q$ is already a signature for $q$,
as it cannot contain $v$ as an endpoint any more, it contains at most $|q|$ walks, but it might not be valid yet.
However, $\sigma'_q$ and $\mathcal{W}_q'$ already satisfy
Conditions~\ref{valid-1},~\ref{valid-2},~\ref{valid-4} from
Definition~\ref{def:valid}.
	I.e, all endpoints of walks are still in $V(q)$, there are at most as many walks as the size of the set of vertices in $V(q)$ used as endpoints, the walks only use the edges they are allowed to.

It is left to satisfy Conditions~\ref{valid-3} (all waypoints are covered),~\ref{valid-5} (all edges specified in the signature are used) and~\ref{valid-6} (optimality) from Definition~\ref{def:valid}.
For Condition~\ref{valid-5}, we can assume that later we adjust $\sigma'_q$ appropriately.
We cover Condition~\ref{valid-3} next:

If $v$ is a waypoint, we do not need to cover it in $\sigma'_q$.
However, the walks $\mathcal{W}_q'$ might not cover all further waypoints.
Consider all the walks $\mathcal{W}_v$ in $\mathcal{W}_b$ not contained in $\mathcal{W}_q'$, minus possibly the walk just consisting of $v$: 
Together with $\mathcal{W}_q'$, they satisfy Condition~\ref{valid-3}, 
but they can use edges incident to $v$ %
and the vertex $v$.
Thus, let $\mathcal{W}_v'$ be the set of walks obtained from $\mathcal{W}_v$ after removing all edges incident to $v$ and the vertex $v$, possibly splitting up every walk into multiple walks.
For every walk (possibly consisting of just a single vertex) in $\mathcal{W}_v'$ holds: 
its endpoints are in $V(q)$.

We now add the walks $\mathcal{W}_v'$ to $\mathcal{W}_q'$, one by one, not violating Conditions~\ref{valid-1},~\ref{valid-2},~\ref{valid-4} (and implicitly,~\ref{valid-5}). After this process, we will also have visited all waypoints, satisfying Condition~\ref{valid-3}.
We start with any walk $W \in \mathcal{W}_v'$:
If $W$ just consists of one vertex $u$, there can be two cases:
First, if $u$ is not an endpoint of a walk $W' \in \mathcal{W}_q'$, then we add $W$ as a walk to $\mathcal{W}_q'$, increasing $\beta_q'$ and $\ell_q'$ from Condition~\ref{valid-2} by one, still holding $|q| \geq \beta_q' \geq \ell_q'$.
Second, if $u$ is an endpoint of a walk $W' \in \mathcal{W}_q'$, we concatenate $W$ and $W'$, keeping $\beta_q'$ and $\ell_q'$ identical.
The case of $W$ being a walk from $u\in V(q)$ to $y \in V(q)$ is similar:
First, if both $u,y$ are not endpoints of walks from $\mathcal{W}_q'$, we add $W$ to $\mathcal{W}_q'$.
Second, if both $u,y$ are endpoints of walks from $\mathcal{W}_q'$, we use $W$ to concatenate them. If the result is a cycle, we pick w.l.o.g.~$u$ as both new endpoints.
Third, if w.l.o.g.~$u$ is an endpoint of a walk $W \in \mathcal{W}_q'$, but $y$ is not, we concatenate $W$, $W'$.

We now obtained $\mathcal{W}_q'$ (and implicitly, $\sigma_q'$) that satisfy Conditions~\ref{valid-1} to~\ref{valid-5} from Definition~\ref{def:valid}, and it is left to show Condition~\ref{valid-6} (optimality).
Assume there is a $\mathcal{W}_{\overline{q}}$ with smaller length than $\mathcal{W}_q'$, both for $\sigma_q'$.
Observe that when reversing our reduction process, the parts of the walks in $G[V_q] \setminus E(G[q])$ are not relevant to our construction, only the signature $\sigma_q'$ as a starting point.
As thus, we can algorithmically (implicitly described in the previous parts of the introduce case) derive all valid solutions and signatures for $b$.

It is left to cover the runtime:
 For every possible signature of the child ($2^{O(\bagsize^2)}$ many), we combine them with every possible edge set ($2^{O(|b|^2)}$ combinations). Then, like unique balls (edges) into bins (walks), we distribute the edges over the walks, also considering all $O(|b|)$ combinations with empty walks, in $|b|^{O(|b|^2)}$ combinations. For every walk, we now obtained an edge set that has to be incorporated into the walk, where we can check in time $O(|b|^2)$ if it is possible and also what the new endpoints have to be (possibly switching both). If the walk is closed, we can pick $O(|b|)$ different endpoints.
All these factors, also the signature size and the number of signatures, are dominated by $|b|^{O(|b|^2)}$, with $|b| \geq 2$.
\end{proof}

\begin{lemma}[Join nodes]\label{lem:join}
Let $b$ be a join node in the nice tree decomposition $\mathcal{T}=(T,X)$, with the two children $q_1 = \text{child}(b)$ and $q_2 = \text{child}(b)$, where we have all valid signatures for $q_1$ and $q_2$.
Then, in time 
$n^{O(|b|)} \cdot 2^{O(\bagsize^2)}$
we can find all the valid signatures of $b$. 
\end{lemma}
Our proof for join nodes consists of two parts, making use of the following fact: For a given valid signature of $b$, two valid sub-solutions with different path traversals have the same total length, if the set of traversed edges is identical.
As thus, when trying to re-create a signature of $b$ with a valid sub-solution, we do not need to create this specific sub-solution, but just \emph{any} sub-solution using the \emph{same} set of endpoints and edges. We show:

\begin{enumerate}
	\item We can partition the edges of a valid sub-solution into two parts along a separator, resulting in a valid signature for each of the two parts, where each sub-solution uses exactly the edges in its part.
	\item Given a sub-solution for each of the two parts separated, we can merge their edge sets, and create all possible signatures and sub-solutions using this merged edge set.
\end{enumerate}

\begin{proof}%
Let $\sigma_b$ be a valid non-empty signature of $b$
with edge set $E_b$, with valid sub-solution $\mathcal{W}_b$.
Our task is to show that we obtain $\sigma_b$ from some valid
signatures $\sigma_{q_1}, \sigma_{q_2}$, with valid sub-solutions
$\mathcal{W}_{q_1}, \mathcal{W}_{q_2}$. 

\begin{clm}\label{edgeset1} 
Given valid $\sigma_b$, $\mathcal{W}_b$, then there must be valid $\sigma_{q_1}, \sigma_{q_2}$, $\mathcal{W}_{q_1}, \mathcal{W}_{q_2}$, such that $E(\mathcal W_{q_1}) \cup E(\mathcal W_{q_2}) = E(\mathcal W_b)$. 
\end{clm}

\begin{clm}\label{edgeset2} 
Given valid $\sigma_{q_1}, \sigma_{q_2}$, $\mathcal{W}_{q_1}, \mathcal{W}_{q_2}$, we show that we can create every possible valid signature of $b$ which has a sub-solution of edge set $E(\mathcal W_{q_1}) \cup E(\mathcal W_{q_2})$ in $n^{O(\texttt{tw})}\cdot 2^{O(\bagsize^2)}$.
\end{clm}

\begin{ClaimProof}[Proof of Claim~\ref{edgeset1}]
Arbitrarily partition $E_b$ into some $E_{q_1}$ and $E_{q_2}$. 
Then, consider $E^{\mathcal W}_{q_1} = \left(E(\mathcal{W}_b) \cap E(T_{q_1})\right)\setminus E_{q_2}$ and $E^{\mathcal W}_{q_2} = \left(E(\mathcal{W}_b) \cap E(T_{q_1})\right)\setminus E_{q_1}$, i.e., the edges of the subwalks corresponding to each child, obtained by the arbitrary partition of $E_b$.

For both $E^{\mathcal W}_{q_1}$ and $E^{\mathcal W}_{q_2}$, we now generate valid signatures and sub-solutions, where the edges of the signatures are already given by $E_{q_1}, E_{q_2}$.
W.l.o.g., we perform this task for $E^{\mathcal W}_{q_1}$: Starting at some vertex $v_1 \in V(q_1)$, generate a walk
 by traversing yet unused incident edges, until no more unused incident edges are left, ending at some $v_2 \in V(q_1)$, possibly $v_1=v_2$ and the used edge set may be empty. 

Perform this for all vertices in $V(q_1)$ not yet used as endpoints, possibly generating walks consisting just of a vertex and no edges. 
However, at most $|V(q_1)|$ walks will be generated, as every endpoint of a walk will not be an endpoint for another walk.
Note that if there is any uncovered set of edges of $E^{\mathcal W}_{q_1}$, they form a set of edge-disjoint cycles, where at least one these cycles will share a vertex with some walk, as only vertices in $V(q_1)$ can have odd degree w.r.t.~the edge set.
Then, we can integrate this cycle into a walk, iterating the process until all edges are covered.
Denote the resulting valid signature by $\sigma_{q_1}$ with edge set $E_{q_1}$.
We perform the same for $E^{\mathcal W}_{q_2}$.
Hence, starting from a valid signature $\sigma_b$ with edge set $E_b$, we created two valid signatures $\sigma_{q_1}$ and $\sigma_{q_2}$ for the two children of $t$, with $\mathcal W_{q_1}, \mathcal W_{q_2}$ such that $E(\mathcal W_{q_1}) \cup E(\mathcal W_{q_2}) = E(\mathcal W_b)$.
\end{ClaimProof}

\begin{ClaimProof}[Proof of Claim~\ref{edgeset2}]
Given $\mathcal W_{q_1}, \mathcal W_{q_2}$, we have to construct every possible $\mathcal W_b$ with $E_{\mathcal W_b} = E(\mathcal W_{q_1}) \cup E(\mathcal W_{q_2})$, which in turn implies its valid signature $\sigma_b$.

In order to do so, we create every possible ($2^{O(\bagsize^2)}$ many) signature $\sigma$ of $b$, possibly not valid ones as well.
However, every such signature $\sigma$ can be checked if it can have a sub-solution using the edges $E(\mathcal W_{q_1}) \cup E(\mathcal W_{q_2})$.
We can obtain the answer to this question via a brute-force approach: We assign every edge from $E(\mathcal W_{q_1}) \cup E(\mathcal W_{q_2})$ to one of the endpoint pairs of $\sigma$, with $|E(\mathcal W_{q_1}) \cup E(\mathcal W_{q_2})| \in O(n^2)$.
For a given $\sigma$ with at most $|b|$ walks, the number of possibilities are thus in $O((n^2)^{|b|})=n^{O(|b|)}$.
 For every endpoint pair $(s_i,r_i)$ of the at most $|b|$ walks, we can check in time linear in the number of edges if the assigned edge set can be covered by a walk between $s_i$ and $r_i$: namely, does the edge set form a connected component where all vertices except for $s_i \neq r_i$ have even degree (or, in the case of $s_i = r_i$, do all vertices have even degree)? When we create a signature multiple times, we can keep any sub-solution of minimum length. In total, the runtime is in $2^{O(\bagsize^2)} \cdot n^{O(\texttt{tw})}$.
\end{ClaimProof}

It remains to cover the case of $\sigma_b$ being empty.
By definition, a valid sub-solution to an empty signatures does not traverse any vertex in $V(b)$. As such, the only way to obtain a valid signature $\sigma_b$ in a join is if both valid $\sigma_{q_1}, \sigma_{q_2}$ are empty, with, w.l.o.g., $\mathcal W_{q_2}$ empty too. Then, $\sigma_b = \sigma_{q_1}$, with $\mathcal W_{b} = \mathcal W_{q_1}$.
\end{proof}

\subsection{Putting it All Together}\label{sec:all-together}
We now have all the necessary tools to prove Theorem~\ref{thm:everything}:
\begin{proof}%
\textbf{Dynamically programming a nice tree decomposition.} 
Translating an instance of the Waypoint Routing Problem to an equivalent one with $s=t$ and unit edge capacities only increases the treewidth by a constant amount, see 
Remark~II
and Lemma~\ref{twapprox}.
Although it is NP-complete to determine the treewidth of a graph and compute an according tree decomposition, 
there are efficient algorithms for constant treewidth~\cite{bodlaender1996linear,DBLP:journals/ijfcs/PerkovicR00}.
Furthermore, Bodlaender et al.~\cite{6686186} presented a constant-factor approximation in a time of $O(c^\texttt{tw} n)$ for some $c \in \mathbb{N}$, also beyond constant treewidth: Using their algorithm $O(\log \texttt{tw})$ times (via binary search over the unknown treewidth size), we obtain a tree decomposition of width $O(\texttt{tw})$.
Following~\cite{DBLP:books/sp/Kloks94}, we generate a nice tree decomposition of treewidth $O(\texttt{tw})$ with $O(\texttt{tw} n) \in O(n^2)$ nodes in an additional time of $O(\texttt{tw}^2n) \in O(n^3)$.
The %
 time so far is $O(c^\texttt{tw} n \log \texttt{tw})+O(\texttt{tw}^2n)$ for some $c \in \mathbb{N}$.

We can now dynamically program the Waypoint Routing Problem on the nice tree decomposition in a bottom-up manner, using Lemma~\ref{lem:leaf} (leaf nodes), Lemma~\ref{lem:forget} (forget nodes), Lemma~\ref{lem:introduce} (introduce nodes), and Lemma~\ref{lem:join} (join nodes).
The time for each programming of a node is at most $O(\texttt{tw})^{O(\texttt{tw}^2)}$ or $n^{O(\texttt{tw})} \cdot 2^{O(\texttt{tw}^2)}$, meaning 
that we obtain all valid signatures with valid sub-solutions at the root node $r$, in a combined time of $n^{O(\texttt{tw}^2)}$, specifically:
$$(O(\texttt{tw})^{O(\texttt{tw}^2)} + n^{O(\texttt{tw})} \cdot 2^{O(\texttt{tw}^2)})\cdot O(\texttt{tw} \cdot n) + O(c^\texttt{tw} n\log \texttt{tw})+O(\texttt{tw}^2n).$$

\noindent\textbf{Obtaining an optimal solution.}
If an optimal solution $I$ to the Waypoint Routing Problem exists (on the unified graph with $s=t$), then the traversed edges $E^*$ and vertices $V^*$ in $I$ yield an Eulerian graph $G^*=(V^*,E^*)$.
With each bag in the nice tree decomposition having $O(\texttt{tw})$ vertices, we can now apply (the Eulerian separation) Lemma~\ref{lem:eulerianseparator-new}: There must be a valid signature of the root $r$ whose sub-solution uses exactly the edges $E^*$. 
As thus, from all the valid sub-solutions at $r$, we pick any solution to WRP\xspace with minimum weight, obtaining an optimal solution to the Waypoint Routing Problem. 
\end{proof}

\section{Walking Through Logarithmically Many Waypoints}\label{sec:polylog}
While the Waypoint Routing Problem is generally NP-hard (as we will see below in Section~\ref{sec:nphard}),
we show that a shortest walk through a bounded (not necessarily constant) 
number of waypoints can be computed in polynomial time.
In the following, we describe reductions to shortest 
vertex-disjoint~\cite{thore-soda,DBLP:conf/ipco/Kawarabayashi08}\footnote{The algorithm in \cite{DBLP:conf/ipco/Kawarabayashi08} is for passing through edges, but a standard reduction also allows to use it for passing through vertices. Similarly, an algorithm for passing through vertices can also be used for edges~\cite[p.22]{2ff4490024874b73b698017e96ea9b14}.} cycle problems, where the cycle has to pass through specified vertices.

As we study walks on \emph{capacitated}
networks instead, we first introduce parallel edges.
Interestingly, two edges are sufficient, see Lemma~\ref{lemma:twice}
in the Appendix. 
Similarly, for edge weights $\omega(e)$, we replace every edge $e$ with a path of length $\omega(e)$.
Lastly, to obtain a simple graph with unit edge weights and unit capacities, we place a vertex on every edge, removing all parallel edges while being distance-preserving.

The transformation of the edge-disjoint cycle problem variant into
a vertex-disjoint route problem variant, and accounting for waypoints, 
however requires some additional considerations.
The standard method to apply vertex-disjoint path 
algorithms to the edge-disjoint case, e.g.,~\cite{DBLP:conf/bonnco/NavesS08,Nishizeki2001177,Shiloach:1978:FTD:322047.322048}\footnote{In~\cite{Nishizeki2001177} it is mentioned that not the line graph is taken, but a graph ``\textit{similar to the line graph}''. Furthermore, Zhou et al.~\cite[p.3]{DBLP:journals/algorithmica/ZhouTN00} suggest to ``[replace] \textit{each vertex with a complete bipartite graph}''.}, 
is to take the line graph $L(G)$ of the original graph $G$.
Then, each edge is represented by a vertex (and vice versa), 
i.e., a vertex-disjoint path in the line graph directly translates to 
an edge-disjoint walk (and vice versa), but possibly changing the graph family.
However, the line graph construction raises the question of where to
place the waypoints. For example, consider a waypoint vertex of 
degree $3$, which is transformed into $3$ vertices in the line graph: 
which of these vertices should represent the waypoint? 

For the case of 2 disjoint paths, Bj{\"o}rklund and Husfeldt~\cite[p.214]{thore-icalp} give the following idea: ``\textit{add an edge to each terminal vertex and apply} [the] \textit{Algorithm} [...] \textit{to the line graph of the resulting graph}''.
Their method is sufficient for $s,t$, but for the remaining waypoints, we also need to add extra vertices to the line graph:
To preserve shortest paths, every shortest ``pass'' through the line graph representation of a vertex $v$ should have the same length, no matter if the waypoint was already visited or not.
As thus, we add $\delta(v)$ further vertices to the line graph, one on each edge connecting two edge representations in the line graph, as in Figure~\ref{fig:clique-expansion}.

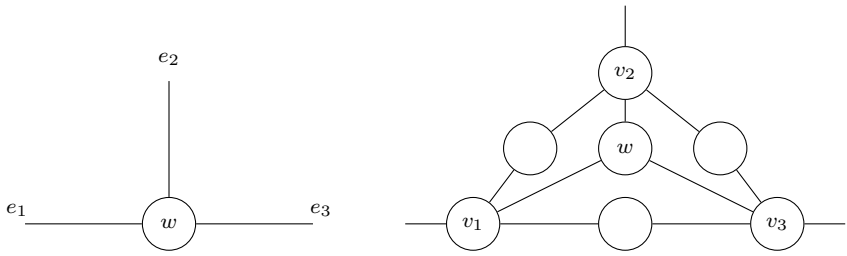
\begin{figure}[t]
	\centering
	\begin{tikzpicture}[auto]
	\node [markovstate] (v) at (0, -1) {$w$};
	\node  (v1) at (-2, -1) {};
	\node  (v2) at (0, 1) {};
	\node  (v3) at (2, -1) {};

	\node [markovstate] (w) at (6, 0) {$w$};
	\node [markovstate] (wum) at (6, -1) {};
	\node [markovstate] (wul) at (4, -1) {$v_1$};
	\node [markovstate] (wml) at (4.75, 0) {};
	\node [markovstate] (wom) at (6, 1) {$v_2$};
	\node [markovstate] (wmr) at (7.25, 0) {};
	\node [markovstate] (wur) at (8, -1) {$v_3$};
	
	\node (woom) at (6, 2) {};
	\node (wull) at (3, -1) {};
	\node (wurr) at (9, -1) {};

	\draw  (v) to (v1) node [above] {$e_1$};
	\draw  (v) to (v2) node [above] {$e_2$};
	\draw  (v) to (v3) node [above] {$e_3$};
	
	\draw  (w) to (wom);
	\draw  (w) to (wur);
	\draw  (w) to (wul);
	
	\draw  (wul) to (wum);
	\draw  (wum) to (wur);
	
	\draw  (wul) to (wml);
	\draw  (wml) to (wom);
	
	\draw  (wom) to (wmr);
	\draw  (wmr) to (wur);
	
	\draw (woom) to (wom);
	\draw (wul) to (wull);
	\draw (wur) to (wurr);

	\end{tikzpicture}
	\caption{Replacing a waypoint vertex in 
	$G$ with an expanded clique in an extended line graph.}
	\label{fig:clique-expansion}
\end{figure}
Next, recall that the algorithm by Bj{\"o}rklund et al.~\cite{thore-soda} 
computes cycles, whereas in WRP\xspace, we are interested in walks
from $s$ to $t$, where $s$ may not equal $t$.
However, due to Remark II (Lemma~\ref{lemma:s=t}), we can assume
that $s=t$.

Given this construction, using Bj{\"o}rklund et al.'s shortest simple cycle algorithm,
we have:

\begin{theorem}\label{thm:cycle-results-1}
For a general graph $G$ with polynomial edge weights, a shortest walk through $k \in O(\log n)$ waypoints can be computed in randomized polynomial time, namely $2^kn^{O(1)}$.
\end{theorem}

Similarly, we can also adapt the result by
Kawarabayashi~\cite{DBLP:conf/ipco/Kawarabayashi08} to derive a deterministic 
algorithm to compute feasible (not necessarily shortest) walks:
\begin{theorem}\label{thm:cycle-results-2}
For a general graph $G$ with polynomial edge weights, a walk through $k \in O\left(\left(\log \log n\right)^{1/10}\right)$ waypoints can be computed in deterministic polynomial time. 
\end{theorem}

Our formal proof of the  Theorems~\ref{thm:cycle-results-1} and~\ref{thm:cycle-results-2} will be a direct implication of the upcoming Corollary~\ref{corr:special-line-graph-shortest}, for which in turn we need the following Theorem~\ref{thm:special-line-graph}. %

For our construction, we will use an extended waypoint-aware line graph $L_{R}(G)$ 
construction in Algorithm~\ref{alg:waypoint-line-graph}.
The fundamental idea is as follows: Similar to the line graph, 
we place vertices on the edges, implying that every edge may only 
be used once. Then, the original vertices are expanded into sufficiently large cliques, 
also containing the waypoint, s.t.~any original edge-disjoint walk 
can also be performed by a path through the clique vertices. 
For an illustration of this so-called clique expansion, we refer again to Figure~\ref{fig:clique-expansion}.

\begin{theorem}\label{thm:special-line-graph}
Consider an instance $I$ of WRP\xspace on  $G=(V,E)$.
If an edge-disjoint route $R$ of length $\ell_1$, solving $I$, exists on $G$,
then there is a
 vertex-disjoint path $P$ from $s$ to $t$ through all waypoints of length $5\ell_1$ on $L_{R}(G)$.
Conversely, if such a $P$ of length $\ell_2$ exists on $L_{R}(G)$, then 
there is a $R$ for $I$ of length $\leq \ell_2/5$.
\end{theorem}

\begin{figure}[t]
\noindent\rule[0.5ex]{\linewidth}{0.4pt}\vspace{-0.8cm}
\begin{algorithm}\label{alg:waypoint-line-graph}
\algcaption{Waypoint Line Graph Construction}
\item \textbf{Input}: Graph $G=(V,E)$, with vertices $s,t,w_1,\dots,w_k \in V$.
\item \textbf{Output}: Graph $L_{R}(G)=(L_{R}(V),L_{R}(E))$, with vertices $s,t,w_1,\dots,w_k \in L_{R}(V)$.
\begin{enumerate}
\item Initialize $L_{R}(E)=E$ and $L_{R}(V)=V$, with the same $s,t,w_1,\dots,w_k$ allocation.
\item For each $v \in L_{R}(V)$
\begin{enumerate}
	\item  Order the incident edges arbitrarily, denoting them locally as $e_{1},\dots,e_{\delta(v)}$, where $\delta(v)$ denotes the vertex degree.
\item Replace every vertex $v \in L_{R}(V)$ with a clique of $\delta(v)+1$ vertices, denoted $K_{\delta(v)}(v)$, naming the vertices locally as $v_{1},\dots,v_{\delta(v)}, v'$, setting any $s,t,w_1,\dots$ on $v$ to $v'$.
\item For the edges incident to the original $v \in L_{R}(V)$, connect the corresponding $e_{i}$ to their $v_{i}$, for $1\leq i \leq \delta(v)$.
\end{enumerate}
\item For each $e \in L_{R}(E)$ not contained in any $K_{\delta(v)}(v)$
\begin{enumerate}
	\item Replace $e$ by a path of three edges and two vertices.
\end{enumerate}
\item For each $e \in L_{R}(E)$ contained in any $K_{\delta(v)}(v)$, not incident to any $v'$
\begin{enumerate}
	\item Replace $e$ by a path of two edges and one vertex.
\end{enumerate}
\end{enumerate} 
\end{algorithm}
\vspace{-0.4cm}
\noindent\rule[0.5ex]{\linewidth}{0.4pt}
\vspace{-0.5cm}
\end{figure}

\begin{proof}%
Given a graph with integer edge weights and capacities, we first transform it into a graph $G$ with unit edge weights and capacities, while being distance-preserving.
First, we replace each edge with a capacity $\geq 1$ with two parallel edges of identical weight, cf.~Lemma~\ref{lemma:twice}.
Second, each edge $e$ with a weight of $\omega(e)$ is replaced by a path of length $\omega(e)$, 
which yields a distance-preserved graph with unit capacities.
Third and last, to remove parallel edges, we place a vertex on \emph{every} edge, obtaining the desired graph properties.

We now start with the case that an edge-disjoint route $R$ of 
length $\ell_1$ exists in $G$, solving $I$.
We translate $R$ into a vertex-disjoint path $P$ in $L_{R}(G)$
of length $5\ell_1$ as follows:
First, every edge $e \in E$ is represented by a path of length 
$3$ in $L_{R}(G)$, resulting in a length of $3\ell_1$ if we could pass
through the ``clique-expansions'' $K_{\delta(v)}(v)$ for free. 
Second, observe that all shortest paths through these 
$K_{\delta(v)}(v)$ have a length of $2$ -- with sufficient 
vertex-disjoint paths to represent all crossings through $v \in V$ 
performed by $R$. 
If $v \in R$ contains a waypoint (or $s,t$), we let one of the crossings in $L_{R}(G)$ pass through $v'$.
When starting on $s$ or ending on $t$, the path-length through the respective is $K_{\delta(v)}(v)$ is only 1.
As such, we showed the existence of a path $P$ from $s$ to $t$ through all waypoints in $L_{R}(G)$ with a length of $3x+2(x-1)+1+1=5\ell_1$.

It is left to show that if such a $P$ of length $\ell_2$ exists 
on $L_{R}(G)$, then there is a $R$, solving $I$, of length $\leq \ell_2/5$. 
We can think of $P$ as follows: It starts in some $K_{\delta(v_1)}(v_1)$ on $s$, 
passes through some $K_{\delta(v_2)}(v_2),\dots,K_{\delta(v_r)}(v_r)$ connected by 
paths of length 3, until it ends in some $K_{\delta(v_{r+1})}(v_{r+1})$ on $t$. 
In this chain, the $K_{\delta(v_i)}(v_i)$s do not need to be pairwise disjoint. 
Observe that every of these paths of length 3 between the $K_{\delta(v)}(v)$s in 
$L_{R}(G)$ directly maps to an edge in $G$. By also mapping the $v \in V$ 
to the extended $K_{\delta(v)}(v)$s, we obtain a one-to-one mapping between 
edge-disjoint walks in $G$ and vertex-disjoint paths in $L_{R}(G)$, 
where the extended $K_{\delta(v)}(v)$s are contracted to a single vertex.
Following the thoughts for the first case, we can shorten $P$ to a path 
$P'$ such that every subsequent traversal of a $K_{\delta(v_i)}(v_i)$ in the 
chain only has a length of 2, which is shortest possible, with the paths 
through $K_{\delta(v_1)}(v_1)$ and $K_{\delta(v_{r+1})}(v_{r+1})$ 
having a length of 1 each (except for the case of $r=0$, which means we can set $|P'|=0$). 
Performing the translation of the first case in reverse, we obtain a 
solution $W$ for $I$ of length $|P'|/5 \leq |P|/5=y/5$.
\end{proof}

The statement of Theorem~\ref{thm:special-line-graph} also has 
implications for shortest solutions.
If there is a shortest vertex-disjoint path in $L_{R}(G)$ of 
length $\ell_2'$, but there exists a solution in $G$ of length $x\ell_1'<\ell_2'/5$, then a solution of length less than $\ell_2'$ would also exist in $L_{R}(G)$, a contradiction.

Let us also briefly consider runtime implications.
When modifying the graph $G=(V,E)$ to be a simple graph $G'=(V',E')$ with unit edge capacities and unit weights, let $f(n) \geq 1$ be the largest edge weight $\omega(e)$ in $G$.
It then holds that $|V'|$ and $|E'|$ are each at most $|V|+4|E|f(n)$.
When considering $L_{R}(G')=(V'_{R},E'_{R})$, we obtain an upper bound (by a large margin) of $7|V|^2 +48|V||E|f(n)+96|E|^2(f(n))^2$ for both $|V'_{R}|$ and $|E'_{R}|$, respectively.
We further bound this term from above (again, by a large margin) via $199|V|^4\left(f(n)\right)^2$.
While this bound can be improved by careful inspection, especially in the size of the exponent, it suffices for the purposes of polynomiality.

\begin{corollary}\label{corr:special-line-graph-shortest}
Let $A$ be an algorithm that finds a shortest vertex-disjoint solution for a path from $s$ to $t=s$ through all specified (waypoint) vertices, with the largest edge weight being of size $f(n)$, in a runtime of $a(k,|V|,f(n))$.
 Using the waypoint line graph construction, algorithm $A$ can be used to find a shortest solution to WRP\xspace in a runtime of $\mathfrak a\left(k,199|V|^4\left(f(n)\right)^2,1\right)$.
\end{corollary}

In particular, if $A$ has a runtime of $2^kn^{O(1)}$ to find a cycle through $k$ specified vertices in an $  n$-vertex graph, we obtain a runtime of $2^k(199n^4\left(f(n)\right)^2)^{O(1)}$ for WRP\xspace .
 If $f(n)$ is a constant-value function or a fixed polynomial, this reduces to $2^kn^{O(1)}$ for $n\geq 2$.
If $\mathfrak a$ is a polynomial function w.r.t.~$k,|V|,f(n)$, it will also be a polynomial function in $n$ for the transformed WRP\xspace instances with inputs $k,199n^4\left(f(n)\right)^2,1$, assuming $f(n)$ is a constant-value function or a fixed polynomial, as we can assume $k<n$.

\section{NP-Hardness}\label{sec:nphard}

Given our polynomial-time algorithms to compute shortest walks through arbitrary waypoints 
on bounded treewidth graphs
as well as to compute shortest walks on arbitrary graphs through a bounded number of waypoints, 
one may wonder whether exact polynomial time solutions also
 exist for more general settings. 
In the following, we show that this is not the case: 
in both dimensions (number of waypoints and more general graph families), 
we inherently hit computational complexity bounds.
Our hardness results follow by reduction from a special subclass
of NP-hard Hamiltonian cycle problems~\cite{akiyama1980np,DBLP:conf/cg/Buro00}:

\begin{theorem}\label{thm:undirected-unordered-npc-3}
WRP\xspace is NP-hard for any graph family of degree at most 3,
for which the Hamiltonian Cycle 
problem is NP-hard.
\end{theorem}

\begin{proof}%
Let $G=(V,E)$
be a graph from a graph family of 
degree at most~$3$ 
for which the Hamiltonian cycle problem is NP-hard 
(e.g.,~\cite{akiyama1980np,ARKIN2009582,DBLP:conf/cg/Buro00,DBLP:journals/jal/PapadimitriouV84}), set all edge capacities to~$1$, %
take an arbitrary vertex
$v\in V$, and set $s:=v=:t$. Set $\ensuremath{\mathscr{W}}:=V\setminus\left\{v\right\}$. 
Consider a route $R$ (a feasible walk) which starts and
ends at $v$ and visits all other vertices. We claim that $R$ is a 
Hamiltonian cycle for $G$; on the other hand, it
is clear that if there is a Hamiltonian cycle of $G$ then it satisfies
the requirements of $R$. We start at $v$ and walk along $R$,
directing edges along the way. Every vertex in the resulting graph
has at least one outgoing directed edge and at least one incoming
directed edge. On the other hand, as the edge capacities are $1$, $R$
cannot reuse any edge, so the number of directed
edges on every vertex must be even: in fact, the number of
incoming edges equals the number of outgoing edges. The maximum degree
of $G$ is~$3$, according to the last two observations, every vertex appears in exactly two
edges of the walk $R$. As $R$ induces a connected subgraph and all its vertices are of degree two, we conclude that $R$ is a single
cycle. As $R$ visits all vertices in $G$, thus it is also a
Hamiltonian cycle.
\end{proof}

\begin{comment}
Armed with the result of Theorem~\ref{thm:undirected-unordered-npc-3}, we can now prove Theorem~\ref{thm:npc-poly-wps}.

\begin{proof}[Proof of Theorem~\ref{thm:npc-poly-wps}]
We know from Theorem~\ref{thm:undirected-unordered-npc-3} that the WRP\xspace is NP-complete on general graphs $G=(V,E)$ with $k$ waypoints $\ensuremath{\mathscr{W}}$, for $k<n$. 
Fix some arbitrary $r \in \mathbb{R}_{\geq 1}$, setting $r'=\left\lceil r \right\rceil$.
 For every $G$, we will construct, in polynomial time, a graph $G'=(V',E')$ with the same set of $k$ waypoints $\ensuremath{\mathscr{W}}$ s.t.~$1)$ $k \in O(n^{1/r'})$ and $2)$ the set of solutions for the WRP\xspace on $G$ and $G'$ is identical.

Observe that when setting $G'=\left(V'=V\cup\left\{v'\right\},E'=E \cup \left\{(s,v')\right\}\right)$ with $c\left((s,v')\right)=1$, any walk starting from $s$ can visit $v'$, but never return, as the new edge $(s,v')$ only has unit capacity. I.e., condition $2)$ is satisfied for $G'$. To satisfy condition $1)$ as well, we connect $v'$ to a clique of $|V|^{r'}$ vertices without any waypoints, enforcing $k \in O(|V'|^{1/r'})$.
\end{proof}
\end{comment}

We have the following implication for grid graphs~\cite{ARKIN2009582,DBLP:conf/cg/Buro00,DBLP:journals/jal/PapadimitriouV84} of maximum degree 3, and use similar ideas for the class of 3-regular bipartite planar graphs.

\begin{corollary}\label{thm:npc-special-graphs-extension}
For any fixed constant $r\geq 1$ it holds that WRP\xspace is NP-hard on grid graphs of maximum degree $3$, already for $k \in O(n^{1/r})$ waypoints.
\end{corollary}

\begin{proof} %
Our proof will be a reduction from~\cite{DBLP:conf/cg/Buro00} which shows that 
the Hamiltonian cycle problem 
is NP-hard on grid graphs of maximum degree 3. Our reduction will not change
these properties of the graph.
We also fix some arbitrary $r \in \mathbb{R}_{\geq 1}$, setting $r'=\left\lceil r \right\rceil$.

For simplicity, we restrict ourselves to the grid graphs $\mathcal{G}$ of maximum degree 3  obtained in~\cite{DBLP:conf/cg/Buro00} by Buro's NP-hardness reduction.
Restricting to $\mathcal{G}$ allows us to follow the arguments from~\cite{DBLP:conf/cg/Buro00} to obtain an appropriate embedding in polynomial time.\footnote{As Arkin et al.~\cite{ARKIN2009582} point out, the first NP-hardness proof for grid graphs $\mathcal{G}^{\dagger}$ of maximum degree $3$  is in~\cite{DBLP:journals/jal/PapadimitriouV84}. Following the references given in~\cite{ARKIN2009582}, it is also possible to embed all $G^{\dagger} \in \mathcal{G}^{\dagger}$ in polynomial time.}
For any WRP\xspace on $G \in \mathcal{G}$ with $k<n$ waypoints $\ensuremath{\mathscr{W}}$,
create a grid drawing in the plane. %
From this drawing, from all vertices with the smallest $x$-coordinates,
pick the vertex $v$ with the smallest $y$ coordinate. 
Say $v$ has coordinates $\left(x(v),y(v)\right)$.
By construction, $v$ has at most a degree of two and there are no vertices 
with a smaller $x$-coordinate than $v$.
As such, we can create a vertex $v'$ with coordinates $\left(x(v)-1,y(v)\right)$, 
and connect it to $v$ with an edge of unit capacity.

Observe that the set of solutions for WRP\xspace was not altered: Once $v'$ is visited, 
no walk can ever leave it.
We now extend this idea, creating a path of length $n^r$, placing its vertices at the coordinates $\left(x(v)-2,y(v)\right)$, $\left(x(v)-3,y(v)\right)$,~$\dots$.
Denote this extended graph by $G'=(V',E')$ and observe that its main properties are preserved,
however, $k \in O(|V'|^{1/r'})$.%
\end{proof}

\begin{corollary}\label{thm:npc-special-graphs-extension-1}
For any fixed constant $r\geq 1$ it holds that WRP\xspace is NP-hard on 3-regular bipartite planar graphs, already for $k \in O(n^{1/r})$ waypoints.
\end{corollary}
\begin{proof}%
Our proof will be a reduction from~\cite{akiyama1980np} which shows that the Hamiltonian cycle problem 
is NP-hard on 3-regular bipartite planar graphs, denoted by $\mathcal{G}_3$. 
Observe that we can assume unit edge capacities, without losing the NP-hardness property.
Again, our reduction will not change these properties of the graph, and analogously, we fix some arbitrary $r \in \mathbb{R}_{\geq 1}$, setting $r'=\left\lceil r \right\rceil$.

We can pick any edge $e=(u,w)$ in a graph $G_3 \in \mathcal{G}_3$ and replace the edge with a path of length three and capacity one, denoting the added vertices by $v$ and $v'$. 
The graph is still bipartite and planar, but $v$ and $v'$ violate the 3-regularity.
Notwithstanding, the feasibility of WRP\xspace stays unchanged: The ``capacity'' of the path between $u$ and $w$ via $v$ and $v'$ is still one.
Now, we create two full binary trees $T, T'$ with a roots $T_v,T_{v'}'$, each having $2^{r'}-1$ vertices and $2^{r'-1}$ leaves.
We can connect $T_v$ to $v$ and $T_{v'}'$ to $v'$, preserving bipartiteness and planarity.
A small examplary construction can be found in Figure~\ref{fig:np3}, already illustrating the next construction steps as well.
Observe that all vertices except the leaves of $T,T'$ have a degree of exactly three.

Next, we pick a standard embedding of $T,T'$ in the plane, s.t., w.l.o.g., the leaves $v_1,\dots,v_{2^{r'-1}}$ have coordinates $(0,1),\dots,(0,2^{r'-1})$, similar for the leaves of $T'$ with $(0,2^{r'-1}+1),\dots,(0,2 \cdot 2^{r'-1})$.

We place a vertex $v^m_{(1,2)},\dots,v^m_{(2^{r'-1}-1, 2^{r'-1})}$ 
between each consecutive leaf of 
$T$, $2^{r'-1}-1$ in total, same for $T'$ 
with vertices $v^{m\prime}_{(1,2)},\dots,v^{m\prime}_{(2^{r'-1}-1, 2^{r'-1})}$. 
Next, we connect these $2\cdot 2^{r'-1}-2$ vertices $V^m$ \& $V^{m\prime}$ with each leaf vertex next to them, also
 $v_{2^{r'-1}}$ with $v_1'$ and $v_1$ 
with $v'_{2^{r'-1}}$, forming a cycle $C$ through all leaves and the new $2\cdot 2^{r'-1}-2$ vertices.
All vertices, except the aforementioned $V^m$ \& $V^{m\prime}$ with a degree of 2, have a degree of 3.
Due to the bipartite property, one can color the (former) leaf vertices and $V^m$ \& $V^{m\prime}$, with, e.g., blue and red: W.l.o.g., we pick red for the (former) leaves of $T$ and $V^{m\prime}$, blue for the (former) leaves of $T'$ and $V^m$.
As the last construction step, we connect $V^m$ and $V^{m\prime}$, as follows, inside the inner face of the cycle $C$: First the two outermost vertices, $v^m_{1,2}$ with $v^{m\prime}_{(2^{r'-1}-1, 2^{r'-1})}$, then going analogously inwards, lastly connecting $v^m_{(2^{r'-1}-1, 2^{r'-1})}$ and $v^{m\prime}_{1,2}$.
Denote this extended graph by $G_3'=(V_3',E_3')$ and observe that its main properties are preserved,
however, $k \in O(|V_3'|^{1/r'})$.
\end{proof}
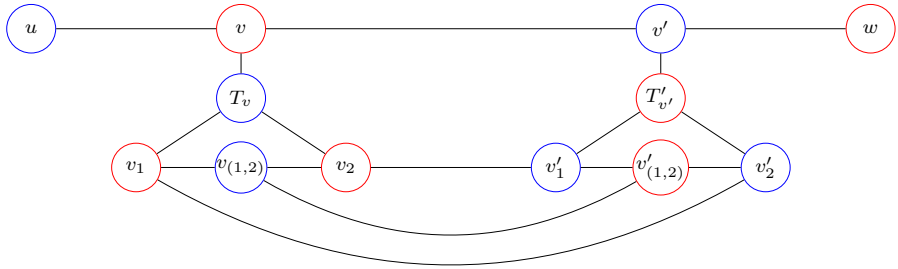
\begin{figure}[t]
\vspace{0.3cm}
\resizebox{\columnwidth}{!}{ 
	\centering
	\begin{tikzpicture}[auto]
	\node [markovstate,draw=blue!200] (ol) at (-3,0) {$u$};
	\node [markovstate,draw=red!200] (or) at (9,0) {$w$};
	
	\node [markovstate,draw=red!200] (v) at (0,0) {$v$};
	\node [markovstate,draw=blue!200] (r) at (0,-1) {$T_v$};
	\node [markovstate,draw=red!200] (v1) at (-1.5,-2) {$v_1$};
	\node [markovstate,draw=red!200] (v2) at (1.5,-2) {$v_2$};
	\node [markovstate,draw=blue!200] (v12) at (0,-2) {$v_{(1,2)}$};
	
	\node [markovstate,draw=blue!200] (v') at (6.0,0) {$v'$};
	\node [markovstate,draw=red!200] (r') at (6.0,-1) {$T_{v'}'$};
	\node [markovstate,draw=blue!200] (v1') at (4.5,-2) {$v_1'$};
	\node [markovstate,draw=blue!200] (v2') at (7.5,-2) {$v_2'$};
	\node [markovstate,draw=red!200] (v12') at (6,-2) {$v'_{(1,2)}$};
	
	\draw (ol) to (v);
	\draw (v) to (v');
	\draw (v') to (or);
	
	\draw (v) to (r);
	\draw (v') to (r');
	
	\draw (r) to (v1);
	\draw (r) to (v2);
	
	\draw (v2) to (v12);
	\draw (v1) to (v12);
	
	\draw (r') to (v1');
	\draw (r') to (v2');

	\draw (v2') to (v12');
	\draw (v1') to (v12');
	
	\draw (v2) to (v1');
	
	\draw (v1) to [bend right] (v2');
	\draw (v12) to [bend right] (v12');
	\end{tikzpicture}
	}
	\caption{Examplary gadget construction with two full binary trees with two leaves each. The resulting graph is 3-regular bipartite and planar.
	}
	\label{fig:np3}
\end{figure}

Our proof techniques also apply to the $k$-Cycle problem studied by, e.g., Bj{\"o}rklund et al.~\cite{thore-soda}, whose solution is polynomial for logarithmic $k$.
All possible edge-disjoint solutions are also vertex-disjoint, due to the restriction of 
maximum degree at most 3.
\begin{corollary}\label{cor:last}
For any fixed constant $r\geq 1$ it holds that the $k$-Cycle problem is NP-hard on $1)$ $3$-regular bipartite planar graphs and $2)$ grid graphs of maximum degree $3$, respectively, already for $k \in O(n^{1/r})$.
\end{corollary}

\section{Conclusion}\label{sec:future}

Motivated by the more general routing models introduced in modern software-defined and function 
virtualized distributed systems, 
we initiated the algorithmic study of computing shortest walks
through waypoints on capacitated networks. We have shown,
perhaps surprisingly, that polynomial-time algorithms exist for a wide range
of problem variants, and in particular for bounded treewidth graphs.

In our dynamic programming approach to the Waypoint Routing Problem, parametrized by treewidth, we provided fixed-parameter tractable (FPT) algorithms for leaf, forget, and introduce nodes, but an XP algorithm for join nodes. 
In fact, while we do not know whether our problem can be expressed
in monadic second-order logic MSO2, we can show that simply concatenating child-walks for join nodes does not result in all valid parent signatures.

\begin{comment}
Moreover, in contrast to existing algorithms to compute feasible and shortest disjoint paths
which have a very high runtime~\cite{}, even on problems for which polynomial-time solutions exist,
we believe that our algorithm is actually practical on bounded treewidth graphs.
\end{comment}

\begin{comment}
In fact, while we currently do not know whether our problem can be expressed
in monadic second-order logic MSO2 (an open question), 
the runtime of our algorithm is relatively fast even compared to
a potential algorithm based on Courselle's theorem.
\end{comment}

We believe that our paper opens an interesting area
for future research. In particular, it will 
be interesting to further chart the complexity landscape of
the Waypoint Routing Problem, narrowing the gap between 
problems for which exact polynomial-time solutions do and
do not exist. Moreover, it would be interesting to derive
a lower bound on the runtime of (deterministic and randomized)
algorithms on bounded treewidth graphs.

\vspace{0.2cm}

\begin{acknowledgements}
The authors would like to thank Riko Jacob for helpful discussions and feedback, as well as the anonymous reviewers of LATIN 2018.
Saeed Akhoondian Amiri's research was partly supported by the European Research
Council (ERC) under the European Union’s Horizon 2020 research and innovation programme
(grant agreement No 648527).
\end{acknowledgements}

\bibliographystyle{spmpsci}      %

\begin{thebibliography}{10}
\providecommand{\url}[1]{{#1}}
\providecommand{\urlprefix}{URL }
\expandafter\ifx\csname urlstyle\endcsname\relax
  \providecommand{\doi}[1]{DOI~\discretionary{}{}{}#1}\else
  \providecommand{\doi}{DOI~\discretionary{}{}{}\begingroup
  \urlstyle{rm}\Url}\fi

\bibitem{akiyama1980np}
Akiyama, T., Nishizeki, T., Saito, N.: {NP}-completeness of the hamiltonian
  cycle problem for bipartite graphs.
\newblock Journal of Information processing \textbf{3}(2), 73--76 (1980)

\bibitem{DBLP:conf/latin/AmiriFS18}
Amiri, S.A., Foerster, K., Schmid, S.: Walking through waypoints.
\newblock In: Proc. {LATIN}, \emph{Lecture Notes in Computer Science}, vol.
  10807, pp. 37--51. Springer (2018)

\bibitem{ifip-waypoint}
Amiri, S.A., Foerster, K.T., Jacob, R., Parham, M., Schmid, S.: Waypoint
  routing in special networks.
\newblock In: Proc. IFIP Networking Conference (2018)

\bibitem{ccr-waypoint}
Amiri, S.A., Foerster, K.T., Jacob, R., Schmid, S.: Charting the algorithmic
  complexity of waypoint routing.
\newblock ACM SIGCOMM Computer Communication Review  (2018)

\bibitem{DBLP:conf/csr/AmiriGKS14}
Amiri, S.A., Golshani, A., Kreutzer, S., Siebertz, S.: Vertex disjoint paths in
  upward planar graphs.
\newblock In: Proc. CSR (2014)

\bibitem{ARKIN2009582}
Arkin, E.M., Fekete, S.P., Islam, K., Meijer, H., Mitchell, J.S.B.,
  Rodr{\'{\i}}guez, Y.N., Polishchuk, V., Rappaport, D., Xiao, H.: Not being
  (super)thin or solid is hard: {A} study of grid hamiltonicity.
\newblock Comput. Geom. \textbf{42}(6-7), 582--605 (2009)

\bibitem{DBLP:journals/dam/ArnborgP89}
Arnborg, S., Proskurowski, A.: Linear time algorithms for np-hard problems
  restricted to partial k-trees.
\newblock Discrete Appl. Math. \textbf{23}(1), 11--24 (1989)

\bibitem{thore-soda}
Bj{\"o}rklund, A., Husfeld, T., Taslaman, N.: Shortest cycle through specified
  elements.
\newblock In: Proc. SODA (2012)

\bibitem{thore-icalp}
Bj{\"o}rklund, A., Husfeldt, T.: Shortest two disjoint paths in polynomial
  time.
\newblock In: Proc. ICALP (2014)

\bibitem{bodlaender1988dynamic}
Bodlaender, H.: Dynamic programming on graphs with bounded treewidth.
\newblock Proc. {{ICALP}}  (1988)

\bibitem{DBLP:journals/actaC/Bodlaender93}
Bodlaender, H.L.: A tourist guide through treewidth.
\newblock Acta Cybern. \textbf{11}(1-2), 1--21 (1993)

\bibitem{bodlaender1996linear}
Bodlaender, H.L.: A linear-time algorithm for finding tree-decompositions of
  small treewidth.
\newblock SIAM Journal on computing \textbf{25}(6), 1305--1317 (1996)

\bibitem{DBLP:journals/iandc/BodlaenderCKN15}
Bodlaender, H.L., Cygan, M., Kratsch, S., Nederlof, J.: Deterministic single
  exponential time algorithms for connectivity problems parameterized by
  treewidth.
\newblock Inf. Comput. \textbf{243}, 86--111 (2015)

\bibitem{6686186}
Bodlaender, H.L., Drange, P.G., Dregi, M.S., Fomin, F.V., Lokshtanov, D.,
  Pilipczuk, M.: An approximation algorithm for treewidth.
\newblock In: Proc. FOCS (2013)

\bibitem{DBLP:journals/algorithmica/BorradaileDT14}
Borradaile, G., Demaine, E.D., Tazari, S.: Polynomial-time approximation
  schemes for subset-connectivity problems in bounded-genus graphs.
\newblock Algorithmica \textbf{68}(2), 287--311 (2014)

\bibitem{DBLP:conf/cg/Buro00}
Buro, M.: Simple amazons endgames and their connection to hamilton circuits in
  cubic subgrid graphs.
\newblock In: Proc. Computers and Games (2000)

\bibitem{chekuri2009note}
Chekuri, C., Khanna, S., Shepherd, F.B.: A note on multiflows and treewidth.
\newblock Algorithmica \textbf{54}(3), 400--412 (2009)

\bibitem{DBLP:books/sp/CyganFKLMPPS15}
Cygan, M., Fomin, F.V., Kowalik, L., Lokshtanov, D., Marx, D., Pilipczuk, M.,
  Pilipczuk, M., Saurabh, S.: Parameterized Algorithms.
\newblock Springer (2015)

\bibitem{cygan2013planar}
Cygan, M., Marx, D., Pilipczuk, M., Pilipczuk, M.: The planar directed
  k-vertex-disjoint paths problem is fixed-parameter tractable.
\newblock In: Proc. FOCS (2013)

\bibitem{DBLP:series/mcs/DowneyF99}
Downey, R.G., Fellows, M.R.: Parameterized Complexity.
\newblock Springer (1999)

\bibitem{eilam1998disjoint}
Eilam-Tzoreff, T.: The disjoint shortest paths problem.
\newblock Discrete applied mathematics \textbf{85}(2), 113--138 (1998)

\bibitem{ene_et_al:LIPIcs:2016:6037}
Ene, A., Mnich, M., Pilipczuk, M., Risteski, A.: {On Routing Disjoint Paths in
  Bounded Treewidth Graphs}.
\newblock In: {{Proc. SWAT}} (2016)

\bibitem{etsi}
{ETSI}: Network functions virtualisation.
\newblock White Paper  (2013)

\bibitem{ETSI1}
{{ETSI}}: Network functions virtualisation (nfv); use cases.
\newblock
  \url{http://www.etsi.org/deliver/etsi_gs/NFV/001_099/001/01.01.01_60/gs_NFV001v010101p.pdf}
  (2014)

\bibitem{sss16moti}
Even, G., Medina, M., Patt-Shamir, B.: Online path computation and function
  placement in {SDNs}.
\newblock In: Proc. SSS (2016)

\bibitem{sirocco16path}
Even, G., Rost, M., Schmid, S.: An approximation algorithm for path computation
  and function placement in {SDN}s.
\newblock In: Proc. SIROCCO (2016)

\bibitem{road}
Feamster, N., Rexford, J., Zegura, E.: The road to {SDN}.
\newblock Queue \textbf{11}(12) (2013)

\bibitem{fellows2007complexity}
Fellows, M., Fomin, F.V., Lokshtanov, D., Rosamond, F., Saurabh, S., Szeider,
  S., Thomassen, C.: On the complexity of some colorful problems parameterized
  by treewidth.
\newblock In: Proc. COCOA. Springer (2007)

\bibitem{minsum2}
Fenner, T., Lachish, O., Popa, A.: Min-sum 2-paths problems.
\newblock Theor. Comp. Sys. \textbf{58}(1), 94--110 (2016)

\bibitem{1991X.1}
Fleischner, H.: {{Eulerian Graphs and Related Topics. Part 1, Volume 2}}, chap.
  {{Chapter X: Algorithms for Eulerian Trails and Cycle Decompositions, Maze
  Search Algorithms}}, pp. X.1 -- X.34.
\newblock North-Holland (1991)

\bibitem{DBLP:journals/ipl/FleischnerW92}
Fleischner, H., Woeginger, G.J.: Detecting cycles through three fixed vertices
  in a graph.
\newblock Inf. Process. Lett. \textbf{42}(1), 29--33 (1992)

\bibitem{Foerster2017}
Foerster, K.T., Parham, M., Schmid, S.: A walk in the clouds: Routing through
  vnfs on bidirected networks.
\newblock In: Proc. ALGOCLOUD (2017)

\bibitem{DBLP:journals/tcs/FortuneHW80}
Fortune, S., Hopcroft, J.E., Wyllie, J.: The directed subgraph homeomorphism
  problem.
\newblock Theor. Comput. Sci. \textbf{10}, 111--121 (1980)

\bibitem{itai1982complexity}
Itai, A., Perl, Y., Shiloach, Y.: The complexity of finding maximum disjoint
  paths with length constraints.
\newblock Networks \textbf{12}(3), 277--286 (1982)

\bibitem{karp1975computational}
Karp, R.M.: On the computational complexity of combinatorial problems.
\newblock Networks \textbf{5}(1), 45--68 (1975)

\bibitem{DBLP:conf/ipco/Kawarabayashi08}
Kawarabayashi, K.: An improved algorithm for finding cycles through elements.
\newblock In: Proc. {IPCO} (2008)

\bibitem{DBLP:journals/siamcomp/KhullerMV92}
Khuller, S., Mitchell, S.G., Vazirani, V.V.: Processor efficient parallel
  algorithms for the two disjoint paths problem and for finding a kuratowski
  homeomorph.
\newblock {SIAM} J. Comput. \textbf{21}(3), 486--506 (1992)

\bibitem{DBLP:journals/siamcomp/KhullerS91}
Khuller, S., Schieber, B.: Efficient parallel algorithms for testing
  k-connectivity and finding disjoint s-t paths in graphs.
\newblock {SIAM} J. Comput. \textbf{20}(2), 352--375 (1991)

\bibitem{DBLP:conf/soda/KleinM14}
Klein, P.N., Marx, D.: A subexponential parameterized algorithm for subset
  {TSP} on planar graphs.
\newblock In: Proc. {SODA} (2014)

\bibitem{DBLP:books/sp/Kloks94}
Kloks, T.: Treewidth, Computations and Approximations, \emph{Lecture Notes in
  Computer Science}, vol. 842.
\newblock Springer (1994)

\bibitem{kobayashi2009shortest7}
Kobayashi, Y., Sommer, C.: On shortest disjoint paths in planar graphs.
\newblock In: Proc. ISAAC (2009)

\bibitem{korte1990paths}
Korte, B., Lovasz, L., Pr{\"o}mel, H.J., Schrijver, L.: Paths, flows, and
  VLSI-layout.
\newblock Springer-Verlag (1990)

\bibitem{DBLP:journals/ipl/Marx04}
Marx, D.: List edge multicoloring in graphs with few cycles.
\newblock Inf. Process. Lett. \textbf{89}(2), 85--90 (2004)

\bibitem{DBLP:conf/bonnco/NavesS08}
Naves, G., Seb{\"{o}}, A.: Multiflow feasibility: An annotated tableau.
\newblock In: W.J. Cook, L.~Lov{\'{a}}sz, J.~Vygen (eds.) Research Trends in
  Combinatorial Optimization, pp. 261--283. Springer (2008)

\bibitem{Nishizeki2001177}
Nishizeki, T., Vygen, J., Zhou, X.: The edge-disjoint paths problem is
  {NP}-complete for series--parallel graphs.
\newblock Discrete Appl. Math. \textbf{115} (2001)

\bibitem{ogier1993distributed}
Ogier, R.G., Rutenburg, V., Shacham, N.: Distributed algorithms for computing
  shortest pairs of disjoint paths.
\newblock IEEE transactions on information theory \textbf{39}(2), 443--455
  (1993)

\bibitem{ohtsuki1981two11}
Ohtsuki, T.: The two disjoint path problem and wire routing design.
\newblock In: Graph Theory and Algorithms, pp. 207--216. Springer (1981)

\bibitem{DBLP:journals/jal/PapadimitriouV84}
Papadimitriou, C.H., Vazirani, U.V.: On two geometric problems related to the
  traveling salesman problem.
\newblock J. Algorithms \textbf{5}(2), 231--246 (1984)

\bibitem{DBLP:journals/ijfcs/PerkovicR00}
Perkovi{\'{c}}, L., Reed, B.A.: An improved algorithm for finding tree
  decompositions of small width.
\newblock Int. J. Found. Comput. Sci. \textbf{11}(3), 365--371 (2000)

\bibitem{DBLP:journals/jct/RobertsonS95b}
Robertson, N., Seymour, P.D.: {{Graph Minors .XIII. The Disjoint Paths
  Problem}}.
\newblock J. Comb. Theory, Ser. {B} \textbf{63}(1), 65--110 (1995)

\bibitem{rost2016service}
Rost, M., Schmid, S.: Service chain and virtual network embeddings:
  Approximations using randomized rounding.
\newblock arXiv preprint 1604.02180  (2016)

\bibitem{e2e}
Saltzer, J.H., Reed, D.P., Clark, D.D.: End-to-end arguments in system design.
\newblock ACM Trans. Comput. Syst. \textbf{2}(4), 277--288 (1984)

\bibitem{scheffler1994practical}
Scheffler, P.: A practical linear time algorithm for disjoint paths in graphs
  with bounded tree-width.
\newblock Technical Report, TU Berlin (1994)

\bibitem{DBLP:journals/siamcomp/Schrijver94}
Schrijver, A.: Finding k disjoint paths in a directed planar graph.
\newblock {SIAM} J. Comput. \textbf{23}(4), 780--788 (1994)

\bibitem{DBLP:conf/focs/SeboZ16}
Seb{\"{o}}, A., van Zuylen, A.: The salesman's improved paths: {A} 3/2+1/34
  approximation.
\newblock In: Proc. {FOCS} (2016)

\bibitem{seymour1980disjoint12}
Seymour, P.D.: Disjoint paths in graphs.
\newblock Discrete Math. \textbf{29}(3), 293--309 (1980)

\bibitem{Shiloach:1980:PSU:322203.322207}
Shiloach, Y.: A polynomial solution to the undirected two paths problem.
\newblock J. ACM \textbf{27}(3), 445--456 (1980)

\bibitem{Shiloach:1978:FTD:322047.322048}
Shiloach, Y., Perl, Y.: Finding two disjoint paths between two pairs of
  vertices in a graph.
\newblock J. ACM \textbf{25}(1), 1--9 (1978)

\bibitem{srinivas2005finding}
Srinivas, A., Modiano, E.: Finding minimum energy disjoint paths in wireless
  ad-hoc networks.
\newblock Wireless Networks \textbf{11}(4), 401--417 (2005)

\bibitem{2ff4490024874b73b698017e96ea9b14}
Taslaman, N.: Exponential-time algorithms and complexity of {NP}-hard graph
  problems.
\newblock Ph.D. thesis, IT-University of Copenhagen, Denmark (2013)

\bibitem{thomassen1980214}
Thomassen, C.: 2-linked graphs.
\newblock European J. Comb. \textbf{1}(4), 371--378 (1980)

\bibitem{verdiere2011shortest3}
de~Verdi{\`e}re, E., Schrijver, A.: Shortest vertex-disjoint two-face paths in
  planar graphs.
\newblock ACM Transactions on Algorithms (TALG) \textbf{7}(2), 19 (2011)

\bibitem{DBLP:journals/algorithmica/ZhouTN00}
Zhou, X., Tamura, S., Nishizeki, T.: Finding edge-disjoint paths in partial
  \emph{k}-trees.
\newblock Algorithmica \textbf{26}(1), 3--30 (2000)

\end{thebibliography}

\appendix

\section{Deferred Claims and Proofs from Section~\ref{subsec:model}}

The idea for the following Lemma~\ref{lemma:twice} can already be found in~\cite[Fig.~1]{DBLP:conf/soda/KleinM14}: %
\begin{lemma}\label{lemma:twice}
Let $R$ be a shortest walk solution to WRP\xspace. 
Then route $R$ visits every edge at most twice. 
\end{lemma}
\begin{proof}
Proof by contradiction.
Construct an edge-weighted multigraph $U$ as follows.
The vertices of $U$ are exactly the vertices of $R$. For every edge $e=\{u,v\}$ 
which appears $x$ times in $R$,
insert edges $e_1,\ldots, e_x$ with endpoints $u,v$ and 
the same weight as $e$, to $U$;
note that there can be parallel edges in $U$. As
$R$ is a walk with start vertex $s$ and end vertex $t$, 
the graph $U$ contains an Eulerian walk which can be obtained 
by following $R$ with respect to the new edges. 
Thus, every vertex in $V(U)\setminus \{s,t\}$ has an even degree 
and $s,t$ have odd degrees. 
For the sake of contradiction, suppose there are two vertices $u,v\in V(U)$ s.t.~there are edges
$e_1,\ldots,e_x$ between them s.t.~ $x > 2$. 
We remove $e_1,e_2$ from $E(U)$ to obtain $U'$. 
The resulting graph is still connected, every vertex except $s,t$ has even degree,
and the graph contains an Eulerian walk $\mathcal{P'}$ which starts at $s$ and ends at $t$. 
But this walk has a smaller total length than $R$, and it also visits all the waypoints. 
This contradicts our assumption that $R$ is a shortest walk.
\end{proof}

\begin{lemma}\label{lemma:s=t}
Consider an instance of WRP\xspace on $G=(V,E)$, with $s \neq t$. By creating a new vertex $v$ connected to both $s,t$, the following two claims hold after setting $v:=s$ and $v:=t$ and creating waypoints on the old placements of $s$ and $t$: $1)$ the treewidth increases at most by one, and $2)$ if and only if there is a shortest solution of length $\ell_1$ in $G$, the shortest solution of the modified WRP\xspace in $G'=(V \cup \left\{v\right\}, E \cup \left\{(s,v),(v,t)\right\}$ is of length $\ell_1+2$.
\end{lemma}

\begin{proof}
We start with the first claim: By placing $v$ into all bags, the treewidth increases at most by one.
For the second claim, we start with the case that there is a shortest solution of length $\ell_1$ in $G$.
Then, we can amend this route in $G'$ to obtain a solution of length $\ell_1+2$. If a shorter solution were to exist in $G'$, we would also obtain a shorter solution in $G$. The reverse case holds analogously, with special coverage of the case that the original WRP\xspace only contains $s,t$ and no further waypoints: Then, due to the placement of waypoints in $G'$ where $s,t$ were placed in $G$, finding a shortest route of length $\ell_1+2$ in $G'$ is equivalent to finding a shortest route of length $\ell_1$ in $G$ of the original WRP\xspace.
\end{proof}

\end{document}